\documentclass[journal,onecolumn,draftclsnofoot]{IEEEtran}
\usepackage[usenames, dvipsnames]{color}
\usepackage{amsmath}
\usepackage{amsbsy}
\usepackage[noend]{algpseudocode}
\usepackage{graphicx}
\usepackage{mathtools}
\usepackage{amssymb}
\usepackage{eqnarray}
\usepackage{multicol}
\usepackage{listings}
\usepackage{color}
\usepackage{units}
\usepackage{amsthm}
\usepackage{amsmath}
\usepackage{vwcol}
\usepackage{algorithm}
\usepackage{footnote}
\usepackage{blkarray}
\usepackage{comment}
\usepackage{verbatim}
\usepackage[sorting=none,style=ieee,backend=biber]{biblatex}
\allowdisplaybreaks

\newtheorem{theorem}{Theorem}
\newtheorem{remark}{Remark}
\newtheorem{corollary}{Corollary}
\newtheorem{lemma}{Lemma}

\newtheorem{assumption}{Assumption}
\newtheorem{proposition}{Proposition}

\newcommand{\cc}{\mid\mid}
\newcommand{\R}{\mathbb{R}}
\newcommand{\ra}{\rightarrow}
\newcommand{\ccdiyxz}{I(Y^n\ra X^n \cc Z^n)}

\newcommand{\mc}[1]{\ensuremath{\mathcal{#1}}}

\newcommand{\gcyx}{\ensuremath{G_{Y\rightarrow X}}}

\newcommand{\diyx}{\ensuremath{I(Y^n\rightarrow X^n)}}

\newcommand{\diyxzm}{\ensuremath{I(Y^{n-1}\rightarrow X^n\mid \mid Z^{n-1})}}

\newcommand{\dirateyx}{\ensuremath{\bar{I}(Y\rightarrow X)}}

\newcommand{\dirateyxz}{\ensuremath{\bar{I}(Y\rightarrow X\mid \mid Z)}}

\newcommand{\diyxm}{\ensuremath{I(Y^{n-1}\rightarrow X^n)}}

\newcommand{\csyx}{\ensuremath{\mathfrak{C}_{Y\rightarrow X}}}

\newcommand{\cxy}{\ensuremath{C_{X\rightarrow Y}}}
\newcommand{\cyx}{\ensuremath{C_{Y\rightarrow X}}}

\newcommand{\cyxp}[1]{\ensuremath{C^{(#1)}_{Y\rightarrow X}}}
\newcommand{\estcxy}{\ensuremath{\hat{C}_{X\rightarrow Y}}}
\newcommand{\estcyx}{\ensuremath{\hat{C}_{Y\rightarrow X}}}

\newcommand{\optcyx}{\ensuremath{C^*_{Y\rightarrow X}}}
\newcommand{\fxr}[1]{\ensuremath{p_{X_#1}^{(r)}}}
\newcommand{\fxc}[1]{\ensuremath{p_{X_#1}^{(c)}}}
\newcommand{\fxp}[2]{\ensuremath{p_{X_#1}^{(#2)}}}

\newcommand{\estfxr}[1]{\ensuremath{\hat{p}_{X_#1}^{(r)}}}
\newcommand{\estfxc}[1]{\ensuremath{\hat{p}_{X_#1}^{(c)}}}
\newcommand{\estfxk}[1]{\ensuremath{\hat{p}_{X_#1}^{(k)}}}

\newcommand{\optfxr}[1]{\ensuremath{p_{X_#1}^{(r)*}}}
\newcommand{\optfxc}[1]{\ensuremath{p_{X_#1}^{(c)*}}}

\newcommand{\history}[1]{\ensuremath{\mc{H}_{#1}}}
\newcommand{\hr}{\ensuremath{\history{i}^{(r)}}}
\newcommand{\hc}{\ensuremath{\history{i}^{(c)}}}

\newcommand{\hp}[1]{\ensuremath{\history{i}^{(#1)}}}
\newcommand{\kl}[2]{\ensuremath{D(#1 \mid \mid #2)}}

\newcommand{\refclass}{\ensuremath{\tilde{\mc{P}}}}
\newcommand{\refclassr}{\ensuremath{\refclass^{(r)}}}
\newcommand{\refclassc}{\ensuremath{\refclass^{(c)}}}
\newcommand{\mcal}[1]{\ensuremath{\mathcal{#1}}}
\newcommand{\absval}[1]{\ensuremath{\left|#1\right|}}
\newcommand{\calx}{\mc{X}}
\newcommand{\caly}{\mc{Y}}
\newcommand{\calz}{\mc{Z}}

\newcommand{\tabp}[2]{\theta^{#1}_{#2}}

\newcommand{\xm}[1]{x_{i-#1}}
\newcommand{\ym}[1]{y_{i-#1}}
\newcommand{\zm}[1]{z_{i-#1}}

\begin{document}

\title{Measuring Sample Path Causal Influences with Relative Entropy}

\author{Gabriel~Schamberg,~\IEEEmembership{Student Member,~IEEE,}
and~Todd~P.~Coleman,~\IEEEmembership{Senior Member,~IEEE,}

\thanks{This work was supported in part by the Center for Science of Information (CSoI), an NSF Science and Technology Center, under grant agreement CCF-0939370; an ARO MURI award under contract ARO-W911NF-15-1-0479; an NIH award 1R01MH110514; and an NSF award IIS-1522125. The material in this paper was presented in part at the 2018 IEEE International Symposium on Information Theory, Vail, CO. and in part at the 2019 IEEE International Symposium on Information Theory, Paris, FR.}%
\thanks{G. Schamberg is with the  Department of Electrical and Computer Engineering, University of California, San Diego, La Jolla, CA 92093, USA (e-mail: gschambe@eng.ucsd.edu).}%
\thanks{T.P. Coleman is with the  Department of Bioengineering, University of California, San Diego, La Jolla, CA 92093, USA (e-mail: tpcoleman@ucsd.edu).}
\thanks{This work has been submitted to the IEEE for possible publication. Copyright may be transferred without notice, after which this version may no longer be accessible.}}

\maketitle

\begin{abstract}
We present a sample path dependent measure of causal influence between time series. The proposed causal measure is a random sequence, a realization of which enables identification of specific patterns that give rise to high levels of causal influence. We show that these patterns cannot be identified by existing measures such as directed information (DI). We demonstrate how sequential prediction theory may be leveraged to estimate the proposed causal measure and introduce a notion of regret for assessing the performance of such estimators. We prove a finite sample bound on this regret that is determined by the worst case regret of the sequential predictors used in the estimator. Justification for the proposed measure is provided through a series of examples, simulations, and application to stock market data. Within the context of estimating DI, we show that, because joint Markovicity of a pair of processes does not imply the marginal Markovicity of individual processes, commonly used plug-in estimators of DI will be biased for a large subset of jointly Markov processes. We introduce a notion of DI with ``stale history'', which can be combined with a plug-in estimator to upper and lower bound the DI when marginal Markovicity does not hold.
\end{abstract}

\begin{IEEEkeywords}
Causality, Relative Entropy, Sequential Prediction, Granger, Markov Process.
\end{IEEEkeywords}

\section{Introduction}

The need to identify and quantify causal relationships between two or more time series from purely observational data is ubiquitous across academic disciplines. Recent examples include economists seeking to understand the directional interdependencies between foreign stock indices \cite{junior2015dependency,
dimpfl2014impact} and neuroscientists seeking to identify directed networks that explain neural spiking activity \cite{kim2011granger,
quinn2011estimating}.

Building upon the the ideas of Wiener \cite{wiener1956theory}, Granger \cite{granger1969investigating} proposed the following perspective on causal influence: \emph{We say that a time series $Y$ is ``causing'' $X$ if we can better predict $X$ given the past of all information than given the past of all information in the universe excluding $Y$}. While Granger's original treatment only considered linear Gaussian regression models, his proposed definition applies in general and is here collectively termed \emph{Granger causality} (GC). The inclusion of ``all information in the universe'' in GC serves to avoid the effects of confounding, i.e. to avoid incorrectly inferring that $Y$ influences $X$ when in reality both $X$ and $Y$ are influenced by a third process, $Z$. It is important to note that GC lacks mention of \emph{interventions}, a concept that is central to well-accepted notions of causal influence popularized by Pearl \cite{pearl2009causality,pearl2018book}. In \cite{eichler2010granger}, Eichler and Didelez develop a framework that formalizes interventions in the context of GC, enabling the distinction between scenarios where changing the value of $Y$ (by means of an intervention) results in a change in the value of $X$, and those where $Y$ merely aids in the \emph{prediction} of $X$. Absent this formal analysis, GC is better viewed as a measure of predictive utility. Nevertheless, it continues to be a popular tool (for example \cite{seth2015granger}), and is the focus of this paper. Thus, we use the terms ``cause,'' ``causal effect,'' etc. within the context of Granger's perspective unless otherwise stated.

More modern information theoretic interpretations of Granger's perspective on causality include directed information (DI) \cite{marko1966theorie,marko1973bidirectional,massey1990causality} and transfer entropy (TE) \cite{schreiber2000measuring}, which is equivalent to GC for Gaussian autoregressive processes \cite{barnett2009granger}. Justification for use of the DI for characterizing directional dependencies between processes was given in \cite{quinn2011equivalence}, where it was shown that, under mild assumptions, the DI graph is equivalent to the so-called minimal generative model graph. It was further shown in \cite{etesami2014directed} that the DI graph can be viewed as a generalization of linear dynamical graphs. As a result, the directional dependencies encoded by DI are well equipped to identify the presence or absence of a causal link under Granger's perspective in the general non-linear and non-Gaussian settings.

Interestingly, both GC and DI are determined entirely by the underlying probabilistic model (i.e. joint distribution) of the random processes in question. It is clear that once the model is determined, these methods provide no ability to distinguish between varying levels of causal influence that may be associated with \emph{specific realizations} of those processes. As a result, GC and DI are only well suited to answer causal questions that are concerned with average influences between processes. Examples of this style of question include \emph{``Does dieting affect body weight?''} and \emph{``Does the Dow Jones stock index influence Hang Seng stock index?''}. Symbolically, we represent this question as Q1:\emph{``Does $Y^{i-1}$ cause $X_i$?''}, where the superscript represents and the collection of samples up to time $i-1$ and capital letters are used to represent random variables and processes.

A natural next question to ask is how the aforementioned measures may be adapted to be sample path dependent. In particular, one might pose the question Q2:\emph{``Did $y^{i-1}$ cause $x_i$?''}, where the lowercase letters now represent specific realizations of the processes $X$ and $Y$. Examples of these questions would be \emph{``Did eating salad cause me to lose weight?''} and \emph{``Did the dip in the price of the Dow Jones cause the spike in the price of the Hang Seng?''}. One information theoretic approach to answering Q2 is the substitution of self-information for entropy wherever entropy appears in the definition of DI \cite{lizier2014measuring}. The issue is that the resulting ``local'' extension of DI may take on negative values, and it is unclear how these values should be interpreted with regard to the presence/absence of a causal link. As a result, causal measures that use the self information have not seen widespread adoption. While this may appear to be a result of a particular methodology, it is in fact a fundamental challenge with Q2 arising from the handling of \emph{counterfactuals}. This challenge relates to what Holland \cite{holland1986statistics} referred to as the ``fundamental problem of causal influence,'' namely that we cannot observe the value that $X_i$ would take under two realizations of $Y^{i-1}$, i.e. the true realization $y^{i-1}$ and some counterfactual realization $\tilde{y}^{i-1}$. A popular approach to dealing with counterfactuals is structural equation models (SEMs). Using an SEM, one can estimate the ``noise'' that gave rise to an outcome $x_i$ and infer the $\tilde{x}_i$ that would have occurred had $y^{i-1}$ been $\tilde{y}^{i-1}$. The interested reader is referred to \cite{pearl2009causality} and \cite{peters2017elements} for more details on SEMs.

While it is clear that Q1 lacks the resolution to identify specific points on a sample path for which a large causal influence is elicited, Q2 introduces the added challenge of counterfactuals and thus there is no clear approach within the GC framework. This observation motivates our proposed question of study, Q3:\emph{``Does $y^{i-1}$ cause $X_i$?''}. In other words, we seek to identify the causal effect that particular values of $Y$ have on \emph{the distribution} of the subsequent sample of $X$. Examples of this include \emph{``Which diets are most informative about weight loss outcomes?''} and \emph{``When does the Dow Jones have the greatest effect on the Hang Seng?''}. To answer this question, we build on the work of \cite{kim2014dynamic} and \cite{schamberg2018sample} in the development of a sample path dependent measure of causal influence.

Such a measure will necessarily capture dynamic changes in causal influence between processes. The means by which causal influences vary with time is two-fold. First, it is clear that when the joint distribution of the collection of processes is non-stationary, there will be variations in time with respect to their causal interactions. Second, we note that stationary processes may exhibit time-varying causal phenomena when certain realizations of a process have a greater level of influence than others (see Section \ref{iid_influences}). The latter cannot be captured by GC and DI, which are determined entirely by the joint distribution and thus will only change when the distribution changes. Furthermore, since estimating GC and DI requires taking a time-average, capturing dynamic changes resulting from non-stationarities necessitates approximating an expectation using a sliding window. The sample path dependent measure, on the other hand, captures both types of temporal dynamics: estimates of the sample path measure can be obtained for any processes for which we can have reliable sequential prediction algorithms.

In developing techniques for estimating the proposed measure, we have identified a challenge in estimating information theoretic measures of causal influence that has been commonly overlooked in the literature. While it is well understood that a collection of jointly Markov processes does not necessarily exhibit Markovicity for \emph{subsets} of processes, the implications of this on information theoretic causal measures are not well studied. An analogous statement with regard to finite order autoregressive processes and the biasing effect this has on estimates of GC was studied in \cite{stokes2017study}, but this work has yet to be adopted in the information theory community. It comes as no surprise that the issues with GC estimators identified in \cite{stokes2017study} may be extended to DI estimators. Thus, a characterization of when estimators of DI are unbiased and a means of addressing the bias when it arises are lacking. As such, we build upon our earlier work \cite{schamberg2019bias} to address both of these unmet needs in Section \ref{inf_order} in an effort to establish an understanding of when one can expect to obtain unbiased estimates of information theoretic causal measures.

The contributions of this paper may be summarized as follows:
\begin{itemize}
\item A methodology for assessing causal influences between time series in a sample-path specific and time-varying manner, by answering the question ``Does $y^{i-1}$ cause $X_i$?''.  This is particularly relevant when there are infrequent events which exhibit large causal influences, which would be ``averaged out'' using any causal measure (e.g. GC and DI) which takes an average over all sample paths.
\item A framework using sequential prediction for estimating the dynamic causal measure with associated upper bounds on the worst case ``causality regret''.
\item A characterization of when unbiased estimates of DI can be obtained using concepts from causal graphical models and a novel methodology for bounding the DI when unbiased estimates cannot be obtained.
\item Demonstration of the causal measure's value through application to simulated and real data.
\end{itemize}

The remainder of this paper is organized as follows: following a brief overview of notation, Section \ref{related} provides a technical summary of related work. In Section \ref{measure} we define the measure, present key properties, and provide justification for the measure through several examples. Section \ref{estimation} provides a framework for estimating the measure. Section \ref{applications} demonstrates the measure on simulated and real data. Finally, Section \ref{discussion} contains a discussion of the results and opportunities for future work.

\subsection{Notation and Definitions}
Let $X$, $Y$, and $Z$ denote discrete finite-alphabet random processes, unless otherwise specified. We denote processes at a given time point with a subscript and denote the space of values they may take with caligraphic letters, i.e. $X_i\in\mc{X}$. Without loss of generality, let $\mc{X}=\{0,1,2,\dots,|\mc{X}|-1\}$.  A temporal range of a process is denoted by a subscript and superscript, i.e. $X_i^n = (X_i,X_{i+1},\dots,X_n)$, and we define $X^n \triangleq X_1^n$. Realizations of processes are given by lowercase letters. Probability mass functions (pmfs) are equivalently referred to as ``distributions'' and are denoted by $f$. These distributions are characterized by a subscript, which is often omitted when context allows. For example $p_{X_i}(x_i) \equiv p(x_i)$ gives the distribution of a single time point of $X$, $p_{X^n,Y^n}(x^n,y^n) \equiv p(x^n,y^n)$ gives the joint distribution of $X$ and $Y$, and $p_{X_i\mid X^{i-1}}(x_i\mid x^{i-1}) \equiv p(x_i\mid x^{i-1})$ gives the conditional distribution of $X$ at a single time conditioned on the past of $X$. Lastly, we define the \emph{causally conditional distribution with lag $k$} as:
\begin{equation}\label{caus_cond_f}
p(x^n\mid\mid y^{n-k}) \triangleq \prod_{i=1}^n p(x_i\mid x^{i-1},y^{i-k}).
\end{equation}

\noindent Note that the standard interpretation of the causal\footnote{The term ``causal'' is overloaded, as it is used here in the control theoretic sense strictly to mean ``non-anticipative.''} conditioning (as in \cite{kramer1998directed}) is recovered by letting $k=0$.

We will briefly review some information theoretic quantities that are used frequently throughout the paper. The entropy is given by:
\begin{equation*}
H(X^n) = \sum_{x^n} p(x^n)\log\frac{1}{p(x^n)}
\end{equation*}
\noindent where it is implied that the sum is over all $x^n \in \mc{X}^n$ and the logarithm is base two (as are all logarithms throughout). The conditional entropy is given by:
\begin{equation*}
H(X^n\mid Y^n) = \sum_{x^n,y^n} p(x^n,y^n)\log\frac{1}{p(x^n\mid y^n)}
\end{equation*}
The causally conditional entropy is given by substituting the causally conditional distribution for the conditional distribution:
\begin{equation*}
H(X^n\mid\mid Y^{n-k}) = \sum_{x^n,y^n} p(x^n,y^n)\log\frac{1}{p(x^n\mid\mid y^{n-k})}
\end{equation*}

For any of the above defined variants of entropy, the corresponding entropy \emph{rates} are given by:
\begin{align}
\bar{H}(X) &= \lim_{n\rightarrow \infty} \frac{1}{n} H(X^n) \\
\bar{H}(X\mid Y) &= \lim_{n\rightarrow \infty} \frac{1}{n} H(X^n\mid Y^n) \\
\bar{H}^{(k)}(X\mid\mid Y) &= \lim_{n\rightarrow \infty} \frac{1}{n} H(X^n\mid\mid Y^{n-k})
\end{align}

\noindent It should be noted that the entropy rates may not exist for all processes.

The conditional mutual information is given by:
\begin{equation*}
I(X^n;Y^n\mid Z^n) = H(X^n\mid Z^n) - H(X^n\mid Y^n, Z^n)
\end{equation*}
\noindent with the (unconditional) mutual information $I(X^n;Y^n)$ being obtained by removing $Z^n$ everywhere it appears in the above equation. Finally, the relative entropy or KL-divergence between two distributions $p_{X_i}$ and $p'_{X_i}$ is given by:
\begin{equation*}
\kl{p_{X_i}}{p'_{X_i}} =
\sum_{x_i} p(x_i)\log\frac{p(x_i)}{p'(x_i)}.
\end{equation*}

\section{Related Work}\label{related}

We now provide a brief summary of three key concepts in the measurement of causal influence across time series, namely Granger causality (GC) \cite{granger1969investigating}, directed information (DI) \cite{marko1973bidirectional,massey1990causality}, and causal strength (CS) \cite{janzing2013quantifying}. While some key points will be presented here, a comprehensive summary of the relationships between GC and DI may be found in \cite{amblard2012relation}.

\subsection{Granger Causality}
While Granger's perspective on causality underlies most modern studies in causality between time series, his original treatment was limited to linear Gaussian AR models \cite{granger1969investigating}. For clarity, we will here present the case with scalar time series. Formally, define the three real-valued random processes $(X_i, Y_i, Z_i: i \geq 1)$. As in Granger's original treatment, we let $Z^n$ represent all the information in the universe in order to avoid the effects of confounding. Next, define two models of $X_i$:
\begin{align}
X_i &= \sum_{j=1}^d
a_j X_{i-j} +
b_j Y_{i-j} +
c_j Z_{i-j} + U_i \label{gc_c} \\
X_i &= \sum_{j=1}^d
d_j X_{i-j} +
e_j Z_{i-j} + V_i \label{gc_r}
\end{align}

\noindent where $a_j,b_j,c_j,d_j,e_j\in\mathbb{R}$ are the model parameters and $U_i\sim \mathcal{N}(0,\sigma_U^2)$ and $V_i\sim \mathcal{N}(0,\sigma_V^2)$. We see that the class of models given by \eqref{gc_r} is a subset of the models given by \eqref{gc_c} where the next $X_i$ does not depend on past $Y^{i-1}$. Thus, a non-negative measure of the extent of causal influence of $Y$ on $X$ may be defined by:
\begin{equation}
\gcyx \triangleq \ln \frac{\sigma^2_V}{\sigma^2_U}
\end{equation}

The limitations of Granger causality extend considerably beyond the restriction to linear models (see \cite{stokes2017study} for a comprehensive summary). Of particular interest is the fact that if a VAR process is of finite order, subsets of the process will in general be infinite order. While it is possible to redefine the model in \eqref{gc_r} to be infinite order, this creates obvious challenges in attempting to estimate Granger causality. Considering this issue is not addressed by the subsequent existing methods, we will revisit this issue in Section \ref{inf_order}.

\subsection{Directed Information}
The concept directed information was first introduced under the name transinformation by Marko in 1973 \cite{marko1973bidirectional} in the context of bidirectional communication theory. It was later revisited in 1990 by Massey \cite{massey1990causality} who defined the directed information from a sequence $Y^n$ to $X^n$ as:
\begin{align}
\diyx
&= \sum_{i=1}^n I(Y^{i};X_i \mid X^{i-1}) \\
&= H(X^n) - H(X^n \mid\mid Y^{n})
\end{align}

Unless otherwise specified, we assume that there are no instantaneous causations, i.e. that $X_i$ and $Y_i$ are conditionally independent given the past $X^{i-1}$ and $Y^{i-1}$. Should one want to allow for instantaneous causations, the proposed methods may be trivially extended to accommodate. As such, we will primarily consider the reverse DI \cite{jiao2013universal}:
\begin{align}
\diyxm
&= \sum_{i=1}^n I(Y^{i-1};X_i \mid X^{i-1})
\label{dinoinst} \\
&= H(X^n) - H(X^n \mid\mid Y^{n-1}) \label{dimentr}
\end{align}

\noindent noting that under the assumption of no instantaneous causation, $\diyx=\diyxm$.

Given that the DI is given by a sum over time, one may be interested in how causal relationships are exhibited on average. This can be accomplished through use of the directed information rate \cite{kramer1998directed}, given by:
\begin{align}
\dirateyx &=\lim_{n\rightarrow \infty} \frac{1}{n}\diyxm \label{dirate} \\
&=\bar{H}(X)-\bar{H}^{(1)}(X\mid\mid Y) \label{dimrateentr}
\end{align}

Lastly, we note that in order to avoid confounding, we may include the side information $Z^n$ to get the following definitions of causally conditioned DI and causally conditioned DI rate:
\begin{gather}
\diyxzm =
\sum_{i=1}^n I(Y^i;X_i \mid X^{i-1},Z^{i-1}) \\
\dirateyxz =
\lim_{n\rightarrow\infty} \frac{1}{n} \diyxzm
\end{gather}

\subsection{Causal Strength}
In \cite{janzing2013quantifying}, Janzing et al. propose an axiomatic measure of \emph{causal strength} (CS) based on a set of postulates that they propose should be satisfied by a causal measure. Furthermore, they present numerous examples to illustrate where Granger causality and directed information do not give results consistent with intuition. While this measure was proposed to measure influences in general causal graphs, it has a clear interpretation in the context of measuring causal influences between two time series. In particular, for measuring the CS from $Y$ to $X$, begin by considering the generalization of the two models utilized by GC in \eqref{gc_c} and \eqref{gc_r} to arbitrary probability distributions $p(X_i \mid X^{i-1},Y^{i-1},Z^{i-1})$ and $p(X_i \mid X^{i-1},Z^{i-1})$. Next, note that the second distribution has the following factorization when summing over all possible pasts of $Y$:
\begin{equation*}
p (X_i \mid X^{i-1},Z^{i-1}) =\sum_{y^{i-1}}
p(X_i \mid X^{i-1},y^{i-1},Z^{i-1})
p(y^{i-1} \mid X^{i-1},Z^{i-1})
\end{equation*}

\noindent The first term in the sum may be viewed as measuring the \emph{direct} effects of the pasts of $X$, $Y$, and $Z$, on the distribution of $X_i$. The second term, however, is in some sense measuring the \emph{indirect} effects of the pasts of $X$ and $Z$ on $X_i$ in that they affect the distribution of $X_i$ through their effect on the distribution of $Y^{i-1}$. Thus, the key idea behind CS is the introduction of the ``post-cutting'' distribution, where the conditional distribution found in the second term is replaced with a marginal distribution (see Section 4.1 of \cite{janzing2013quantifying} for a formal definition). As a result, the (time series) CS from $Y$ to $X$ with side information $Z$ is given by:
\begin{equation}
\csyx \triangleq E\left[D(p_{X_i\mid X^{i-1},Y^{i-1},Z^{i-1}} \mid \mid \tilde{p}_{X_i\mid X^{i-1},Z^{i-1}})\right]
\end{equation}

\noindent where the expectation is taken with respect to $p_{X^{i-1},Y^{i-1},Z^{i-1}}$ and the post-cutting distribution is defined as:
\begin{equation}
\tilde{p}_{X_i\mid X^{i-1},Z^{i-1}}(X_i)
\triangleq
\sum_{y^{i-1}}
p(X_i \mid X^{i-1},y^{i-1},Z^{i-1})
p(y^{i-1})
\end{equation}

\noindent The post-cutting distribution is designed to ensure that the extent to which $Y$ has a causal effect on $X$ depends only upon $Y$ and other \emph{direct} causes of $X$ (see P2 in \cite{janzing2013quantifying}). In the context of measuring causal influences between time series, this can be seen as correcting for scenarios in which $X$ may be very well \emph{predicted} by its own past while not being \emph{caused} by its own past. This scenario arises in models like the one depicted in the center of Figure \ref{fig:examplegraphs}. In such a scenario, it is possible to have $\diyxm=0$ despite the fact that $Y_{i-1}$ is, in some sense, the sole cause of $X_i$. The details of this example are made clear in Section \ref{xcopy}.

By presenting an axiomatic framework for measuring causal influences, Janzing et al. provide a robust justification CS. With that said, we note that like GC and DI, CS is determined solely by the underlying probabilistic model. As such, it may be the preferred technique for addressing Q1, but it does not represent how different realizations may give rise to different levels of causal influence.

\subsection{Self-Information Measures}
All of the aforementioned techniques involve taking an expectation over the histories of the time series in question, and are thus well suited to address Q1. In order to address Q2, a notion of locality may be introduced through use of \emph{self-information}. For a given realization $x$ of a random variable $X\sim p_X$, the self-information is given by $h(x) \triangleq -\log p_X(x)$ and represents the amount of surprise associated with that realization. By replacing entropy with self-information, and its conditional form $h(x\mid y) \triangleq -\log p_{X\mid Y}(x\mid y) $, a local version of DI and its conditional extension may be obtained (see Table 1 in \cite{lizier2014jidt} for other so-called ``local measures''). As an example, we note that for a given pair of realizations $x^i$ and $y^{i-1}$, a ``directed information density'' (using the language of \cite{gray2011entropy}) may be given by:
\begin{equation}
i(y^{n-1}\rightarrow x^n) = \sum_{i=1}^n \log\frac{p(x_i\mid x^{i-1},y^{i-1})}{p(x_i\mid x^{i-1})}
\end{equation}

\noindent While this indeed creates a sample path measure of causality whose expectation is DI, it is clear it may take on negative values. Such a scenario occurs when the knowledge that $Y^{i-1}=y^{i-1}$ makes the observation of $X_i=x_i$ less likely to have occurred. While self-information measures are a good candidate for beginning to address Q2 given their dependence upon realizations, the potential for negative values creates difficulty in trying to obtain an easily interpretable answer in all cases.

\subsection{Time-Varying Causal Measures}
A popular extension of GC style causal measures is application to time-varying scenarios \cite{sheikhattar2018extracting,oselio2017dynamic}. In order to adapt existing methods to these types of scenarios, it is necessary to evaluate them over stretches of time for which there is stationarity. As such, estimation in this scenario necessitates some sort of sliding window technique in order to approximate an expectation, giving rise to a trade-off between sensitivity to dynamic changes and accuracy. Despite being concerned with time-varying causal influences, these approaches are still ultimately attempts to answer Q1 in that the quantity being estimated is determined solely by the underlying joint distribution. The temporal variability that is measured by these approaches is a result only of potential non-stationarities. This is fundamentally different from the question we are asking, which is concerned with the dynamic causal influences that are associated with a particular realization of a process that may or may not be stationary.

\section{A Sample Path Measure of Causal Influence}\label{measure}

%

We begin by considering the scenario where, having observed $(x^{i-1},y^{i-1},z^{i-1})$, we wish to determine the causal influence that $y^{i-1}$ has on the next observation of $X_i$. Define the \emph{restricted} (denoted $(r)$) and \emph{complete} (denoted $(c)$) histories as:
\begin{align*}
\hr &\triangleq \{x_1,\dots,x_{i-1}\}  \cup \{z_1,\dots,z_{i-1}\} \\
\hc &\triangleq \hr \cup \{y_1,\dots,y_{i-1}\}
\end{align*}

\noindent The current time samples of side information from the histories (i.e. $y_i$ and $z_i$) are intentionally omitted, as we assume that there is no instantaneous coupling. We next define the restricted and complete conditional distributions as:
\begin{align*}
\fxr{i}(x_i) & \triangleq p_{X_i}
    (x_i \mid \hr) \\
\fxc{i}(x_i) & \triangleq p_{X_i}
    (x_i \mid \hc).
\end{align*}

\noindent Using these distributions, the sample path causal measure from $Y$ to $X$ in the presence of side information $Z$ at time $i$ is defined by:
\begin{equation}
\cyx(\hc) \triangleq \kl{\fxc{i}}{\fxr{i}} \label{Cdef}
\end{equation}

\noindent For ease of notation, we may refer to the causal measure at time $i$ simply as $\cyx(i)$.

The proposed causal measure has an interesting relationship to the directed information. To illustrate this, consider the conditional mutual information term that appears in the sum in \eqref{dinoinst}, along with two equivalent representations:
\begin{align}
I(Y^{i-1};X_i \mid X^{i-1}) &= H(X_i \mid X^{i-1})-H(X_i \mid X^{i-1},Y^{i-1}) \label{di_entropy}\\
&= E_{X^i,Y^{i-1}}\left[
\log \frac{p(X_i \mid X^{i-1},Y^{i-1}) }{p(X_i \mid X^{i-1})}
\right] \label{di_logratio}
\end{align}

\noindent These equivalent definitions of directed information yield two interpretations. While \eqref{di_entropy} considers the reduction in uncertainty obtained by conditioning on $Y^{i-1}$, \eqref{di_logratio} considers the change in the distribution resulting from the added conditioning as measured by a log-likelihood ratio. When we wish to condition on a realization $(X^{i-1},Y^{i-1})=(x^{i-1},y^{i-1})$, these representations are no longer equivalent:
\begin{gather}
H(X_i \mid x^{i-1})-H(X_i \mid x^{i-1},y^{i-1}) \label{cond_entropies}\\
\ne \nonumber \\
E_{X_i}\left[
\log \frac{p(X_i \mid x^{i-1},y^{i-1}) }{p(X_i \mid x^{i-1})}
\middle|
x^{i-1},y^{i-1}\right] \label{causal_logratio}
\end{gather}

\noindent The representation given by \eqref{causal_logratio} is chosen to be the sample path causal measure and is indeed equivalent to the proposed measure in \eqref{Cdef}. This choice is made clear by noting two properties of \eqref{cond_entropies}. First, we note that \eqref{cond_entropies} may be negative. Second, for particular realizations of $x^{i-1}$ and $y^{i-1}$, we may have that conditioning on $y^{i-1}$ drastically shifts the distribution of $X_i$ while only mildly affecting the conditional entropy, yielding a value of nearly zero for a scenario when there is a clear causal influence. We note that the difference between definitions of DI that is induced by conditioning on a realization is acknowledged in \cite{jiao2013universal}, where four unique estimators of DI are proposed based on these various equivalent definitions of DI. While these estimators converge to the same result in the estimation of DI, the different perspectives yield different results for the question we are addressing and thus their implications must be considered.

As a result of the added conditioning, the proposed measure is a \emph{random variable} that takes on a value for each possible history and may be related to the directed information as follows:

\begin{proposition}\label{prop:te}
In the absence of instantaneous influences, the sum of the expectation over sample paths of the proposed causal measure is the directed information:
\begin{equation}\label{teprop}
\sum_{i=1}^n E
[\cyx(\hc)] = \diyxzm
\end{equation}
\end{proposition}

\noindent See Appendix \ref{app:te} for a proof of the proposition.

A second key property of the proposed measure is non-negativity (for any history), which follows directly from the properties of the KL-divergence. Furthermore, the measure will take a value of zero if and only if the complete and restricted distributions are equivalent for a given history. As such, the proposed causal measure may take on a large value when the additional condition on $y^{i-1}$ introduces a large amount of \emph{uncertainty} into the distribution of $X_i$. In such a scenario, we would expect $y^{i-1}$ to have a significant causal influence on $X_i$ even though it is not causing $X_i$ to take on a \emph{specific value}. It is this type of scenario that makes Q2 so difficult to answer in a consistent manner, despite having a clear interpretation in terms of Q3.

\begin{remark}\label{remark:intervention}
Despite the fact that Granger's perspective on causal influence includes no remarks on the role of interventions, it may be of interest to consider a version of the proposed measure where the value of the influencer is forced by means of an intervention\footnote{We thank an anonymous reviewer for this suggestion.}. For example, one might wish to consider:
\begin{equation}
C^{(i)}_{Y\ra X}(\mc{H}_i^{(c)})\triangleq \kl{p_{X_i}^{(i)}}{p_{X_i}^{(r)}}
\end{equation}
\noindent where the first argument of the divergence is defined as the \emph{interventional} distribution:
\begin{equation}
p_{X_i}^{(i)}(x_i) \triangleq p(x_i \mid \mc{H}_i^{(r)},do(Y^{i-1}=y^{i-1}))
\end{equation}
\noindent We have here used the $do$-operator of Pearl \cite{pearl2009causality} to represent the action of forcing $Y^{i-1}$ to take the values $y^{i-1}$ irrespective of the probability with which those values occur. Given that it is often infeasible to perform such interventions in real world scenarios, a large body of causality research is focused on determining when these interventional distributions can be learned from observed data. While providing a general characterization of the scenarios for which $C^{(i)}_{Y\ra X}=C_{Y\ra X}$ is outside the scope of this paper, this equivalence does in fact hold for the three examples in the following section (with a mild technical assumption). This follows intuitively from the depictions in Figure \ref{fig:examplegraphs} and is shown formally in Appendix \ref{app:backdoor}. By contrast, we would not expect this equivalence to hold for the stock market example considered in Section \ref{djhs}, as discussed in Remark \ref{remark:djhs}.
\end{remark}
\subsection{Justification for Measurement of Sample Path Influences}\label{sec:examples}

We now present a series of examples that illustrate the value of a sample path causal measure. Graphical representations of the three examples can be seen in Figure \ref{fig:examplegraphs}.

\begin{figure*}[t!]
  \centering
  \includegraphics[keepaspectratio, width=\textwidth]{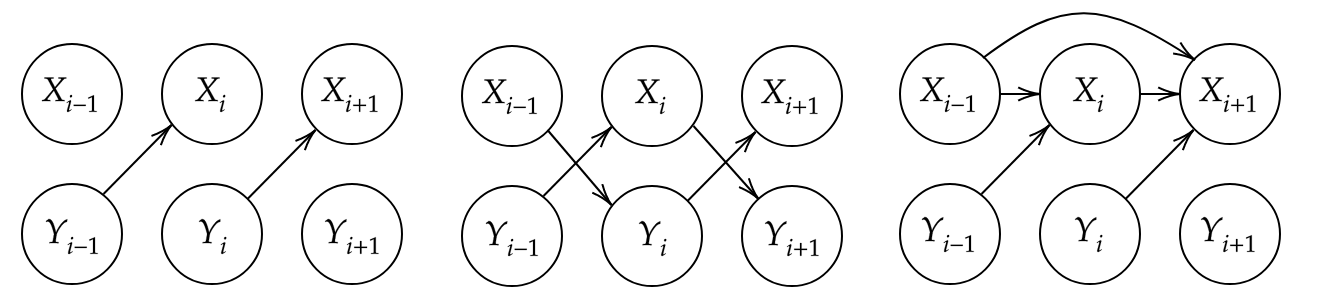}
  \caption{Graphical representation of the IID influences (left), perturbed cross copying (center), and horse betting (right) examples}
  \label{fig:examplegraphs}
\end{figure*}

\subsubsection{IID Influences}\label{iid_influences}
Let $Y_i \sim Bern(\epsilon)$ iid for $i=1,2,\dots$  and:
\begin{equation}\label{ex_one}
X_i \sim
\begin{cases}
      \text{Bern}(p_1), & Y_{i-1} = 1 \\
      \text{Bern}(p_2), & Y_{i-1} = 0
\end{cases}
\end{equation}
\noindent for $\epsilon,p_1,p_2\in[0,1]$. Intuitively, the extent to which $Y_{i-1}$ influences $X_i$ will vary for different values $y_{i-1}$ provided that $p_1\ne p_2$. In order to compute the causal measure $\cyx(i)$, we first need to find the restricted distribution of $X_i$ given only its own past:
\begin{equation*}
\begin{aligned}
\fxr{i}(1)&=\mathbb{P}(X_i = 1 | X^{i-1}=x^{i-1}) \\
&= \mathbb{P}(X_i = 1) \\
&= \sum_{y_{i-1}\in \{0,1\}}
    \mathbb{P}(X_i =1 \mid Y_{i-1} = y_{i-1}) \mathbb{P}(Y_{i-1} = y_{i-1}) \\
&= p_1\epsilon + p_2(1-\epsilon).
\end{aligned}
\end{equation*}
\noindent Noting that $\fxc{i}(1) = p_1$ when $y_{i-1}=1$ and $\fxc{i}(1) = p_2$ when $y_{i-1}=0$, the causal measure is given by:
\begin{equation*}
\cyx(i) =
\begin{cases}
     D(p_1 \mid\mid  p_1\epsilon + p_2(1-\epsilon)),
     & y_{i-1} = 1 \\
     D(p_2 \mid\mid  p_1\epsilon + p_2(1-\epsilon)),
     & y_{i-1} = 0
\end{cases}
\end{equation*}
\noindent Thus, we see that as $\epsilon \rightarrow 0$,
\begin{equation*}
\cyx(i) \rightarrow
\begin{cases}
     D(p_1 \mid\mid  p_2),
     & y_{i-1} = 1 \\
     0,
     & y_{i-1} = 0
\end{cases}
\end{equation*}
\noindent By contrast, the DI rate is given by taking the expectation of $\cyx(i)$ over possible values $Y_{i-1}$. Defining $\cyx(i) \triangleq \cyx(y_{i-1})$, we get:
\begin{equation*}
\dirateyx=
\cyx(1)\epsilon + \cyx(0)(1-\epsilon)
\xrightarrow{\epsilon\rightarrow 0} 0
\end{equation*}
\noindent As a result, it is clear that the sample paths that occur with lower probability will give rise to a greater causal measure than those that occur with higher probability; however, as a result of their lesser probability, these infrequent, highly influential events will have little influence in the computation of the DI rate.

We further note that while it is tempting to invoke ``conditioning reduces entropy'' to conclude that $\cyx(i) > 0$ represents a \emph{reduction} in uncertainty that is obtained by including the past $y^{i-1}$ in the prediction of $X_i$, this is not the case. To make this clear, assign values $p_1 \approx 0.5$ and $p_2 \approx 1$ in \eqref{ex_one} and again let $\epsilon$ approach zero. In such a scenario, we find that:
\begin{equation*}
\fxr{i}(1)\approx 1 \ \ \ \ \ \ \ \ \ \
\fxc{i}(1)\approx
\begin{cases}
1, & y_{i-1}=1 \ (w.p. \ 1-\epsilon)  \\
0.5, & y_{i-1}=0 \ (w.p. \ \epsilon)
\end{cases}
\end{equation*}

\noindent As such, it is clear that by additionally conditioning on $y_{i-1}=0$, there is a considerable increase in uncertainty. Thus, while it is certainly true that $H(X_i \mid X^{i-1}) \le H(X_i \mid X^{i-1},Y^{i-1})$, there are scenarios in which a particular realization of $Y^{i-1}$ may \emph{cause uncertainty} in $X_i$. Revisiting Q2, it is not clear how to answer the extent to which the event $\{Y_{i-1}=0\}$ causes any particular outcome $\{X_i=x_i\}$, because all possible outcomes are equally likely. On the other hand, if we consider Q3, it is quite clear that the event $\{Y_{i-1}=1\}$ has significant influence on $X_i$ and that this is reflected by the proposed measure.

\subsubsection{Perturbed Cross Copying}\label{xcopy}
We next consider a scenario where two processes repeatedly swap values. This example was originally posed in \cite{ay2008information} and modified to include noise in \cite{janzing2013quantifying}. Formally, the processes may be defined as:
\begin{equation}\label{ex_two}
X_i =
\begin{cases}
Y_{i-1}, \ w.p. \ 1 -\epsilon \\
Y_{i-1} \oplus 1, \ w.p. \ \epsilon
\end{cases} \ \ \ \ \
Y_i =
\begin{cases}
X_{i-1}, \ w.p. \ 1 -\epsilon \\
X_{i-1} \oplus 1, \ w.p. \ \epsilon
\end{cases}
\end{equation}

\noindent where $X_i,Y_i\in\{0,1\}$ for all $i$ and $\oplus$ is the XOR operator. We again consider the limiting case where $\epsilon$ is taken to approach zero. As is shown in \cite{janzing2013quantifying}, the DI rate approaches zero as $\epsilon \rightarrow 0$. This results from the fact that for very small $\epsilon$, $Y_{i-1}$ on average contains virtually no information about $X_i$ that is not contained in $X_{i-2}$.

Janzing et al. \cite{janzing2013quantifying} note that because $X_i$ and $X_{i-2}$ are \emph{independent} given $Y_{i-1}$, $Y_{i-1}$ should, in some sense, be fully responsible for the information that is known about $X_i$. As a result, for this example their proposed causal strength measures the average reduction in uncertainty obtained by conditioning on $Y_{i-1}$ versus conditioning on \emph{nothing at all}, i.e. $\csyx = D(\epsilon \mid\mid 0.5)\rightarrow 1$ as $\epsilon \rightarrow 0$ (under the assumption the $X$ and $Y$ are initiated by fair coin tosses).

Next, we consider our proposed sample path measure. First, we note that the complete distribution of $X_i$ depends only upon $y_{i-1}$ and the restricted distribution depends only upon $x_{i-2}$. Explicitly, we get the following distributions:
\begin{equation*}
\fxc{i}(x_i) =
\begin{cases}
1-\epsilon, & x_i=y_{i-1} \\
\epsilon, & x_i\ne y_{i-1}
\end{cases} \ \ \ \ \ \ \ \ \
\fxr{i}(x_i) =
\begin{cases}
\epsilon^2 + (1-\epsilon)^2, & x_i=x_{i-2} \\
2\epsilon(1-\epsilon), & x_i\ne x_{i-2}
\end{cases}
\end{equation*}

\noindent As a result, we see that for a given complete history $\hc= \{x_{i-2},y_{i-1}\}$ we get:
\begin{equation*}
\cyx(\hc)=
\begin{cases}
D(\epsilon \mid\mid 2\epsilon(1-\epsilon)), & x_{i-2}=y_{i-1} \\
D(\epsilon \mid\mid \epsilon^2 + (1-\epsilon)^2), & x_{i-2}\ne y_{i-1} \\
\end{cases}
\end{equation*}

\noindent Thus, we see that as $\epsilon\rightarrow 0$, $\cyx \rightarrow 0$ if $x_{i-2}=y_{i-1}$ and $\cyx \rightarrow \infty$ otherwise.

A comparison of the three measures makes clear that each provides a slightly different perspective. DI rate is loyal to the Granger's perspective in that it captures how, as $\epsilon \rightarrow 0$, $Y_{i-1}$ contains less and less information about $X_i$ that is \emph{not already known}. As a result $\dirateyx$ is strictly decreasing for decreasing $\epsilon$. Causal strength, on the other hand, is loyal to the causal Markov condition in the sense that it restricts consideration to only the immediate parents of the node in question (see P2 in Section 2 of \cite{janzing2013quantifying}). As such, decreasing $\epsilon$ yields a smaller level of uncertainty in $X_i$ conditioned on $Y_{i-1}$, and therefore the causal strength is strictly increasing for decreasing $\epsilon$. The proposed measure lies somewhere in between the two in that it simultaneously captures the decrease and increase in effect of $Y$ on $X$ as $\epsilon$ shrinks. Deciding which perspective is ``correct'' is a philosophical question that must be answered on a problem-by-problem basis. In any case, the proposed measure provides an interesting perspective that, to our knowledge, has not been considered in the literature.

\subsubsection{Horse Betting}\label{betting}
Consider the problem of horse race gambling with side information as presented in Section III-A of \cite{permuter2011interpretations} (with minor adjustments to notation). At each time $i$ the gambler bets all of their wealth based on the past winners $X^{i-1}\in [M]^{i-1}$ and side information $Y^{i-1}$. As a result, the gambler's wealth at time $i$, denoted $w(X^{i},Y^{i-1})$, is a function of the winning horses and side information up to that time. Lastly, the amount of money that is won for betting on the winning horse is given by the odds $o(X_i \mid X^{i-1})$, and the portion of wealth bet on each horse is given by $b(X_i\mid X^{i-1},Y^{i-1}) \ge 0$ with $\sum_x b(x\mid  X^{i-1},Y^{i-1}) = 1$ . Thus, the evolution of the wealth can be described recursively as:
\begin{equation*}
w(X^{i},Y^{i-1}) = b(X_i\mid X^{i-1},Y^{i-1})o(X_i\mid X^{i-1})w(X^{i-1},Y^{i-2})
\end{equation*}
\noindent Finally, the expected growth rate of the wealth is defined as $\frac{1}{n}E[\log w(X^n,Y^{n-1})]$.

It is shown in \cite{permuter2011interpretations} that the betting strategy that maximizes the expected growth rate is given by distributing bets according to the conditional distribution of $X_i$ given all available information:
\begin{equation*}
b^*(X_i \mid X^{i-1},Y^{i-1}) = p(X_i \mid X^{i-1},Y^{i-1}).
\end{equation*}

Similarly, we can define a restricted betting strategy $b(X_i \mid X^{i-1})$ where the side information is not available (and optimal strategy $b^*(X_i \mid X^{i-1}) = p(X_i \mid X^{i-1})$). The wealth that is obtained under that strategy is then given by:
\begin{equation*}
w(X^i) =
b(X_i\mid X^{i-1})o(X_i\mid X^{i-1})w(X^{i-1})
\end{equation*}

\noindent Letting $w^*(X^{i},Y^{i-1})$ and $w^*(X^{i})$ represent the wealth resulting from using the optimal strategies, it is further shown in \cite{permuter2011interpretations} that the increase in growth rate resulting from including side information in the betting strategy is given by:
\begin{equation}\label{di_betting}
\frac{1}{n}E\left[\log w^*(X^{n},Y^{n-1})-\log w^*(X^{n})\right] =  \frac{1}{n}\diyxm
\end{equation}
It should be noted that the result in \eqref{di_betting} holds for any choice of odds $o(X_i\mid X^{i-1})$. Thus, we proceed by making the mild assumption that the odds chosen by the racetrack are such that, for any past sequence of winners $x^{i-1}$, the gambler optimally betting without side information is expected to lose money on round $i$:
\begin{equation}\label{losers}
E[\log b^*(X_i \mid X^{i-1})o(X_i\mid X^{i-1}) \mid x^{i-1}] = \log \delta < 0
\end{equation}
\noindent for some $0< \delta <1$. We define the above equation as the conditional expected growth rate for race $i$ (without side information). As a consequence, this implies a negative expected growth rate for the gambler's wealth without side information:
\begin{align*}
E[\log w^*(X^n)] &=E[\log b^*(X_n \mid X^{n-1})o(X_n\mid X^{n-1})]+E[\log w^*(X^{n-1})]\\
&=\sum_{i=1}^n E[\log b^*(X_i \mid X^{i-1})o(X_i\mid X^{i-1})]+\log w_0 \\
&=n \log \delta < 0
\end{align*}
\noindent where the initial wealth $w_0$ is assumed, without loss of generality, to be $1$.

It follows that a gambler with access to side information ought to gamble only if their expected growth rate is greater than zero. Applying this condition to \eqref{di_betting}, a gambler with side info can expect to win money if:
\begin{equation}\label{winner}
\frac{1}{n}\diyxm > -\log \delta
\end{equation}

Thus, when equipped with the DI, a gambler will decide either to visit the racetrack and bet on every race or to stay at home. It turns out, however, that the gambler may be doing themselves a disservice by staying home any time that \eqref{winner} does not hold. To see this, suppose that before race $i$ the gambler has witnessed winners $x^{i-1}$ and side information $y^{i-1}$, and wishes to gamble if they expect to make money on the current race. Such a scenario occurs when the conditional expected growth rate for round $i$ is positive:
\begin{equation}\label{roundwinner}
E[\log b^*(X_i \mid X^{i-1},Y^{i-1})o(X_i\mid X^{i-1}) \mid x^{i-1},y^{i-1}] > 0
\end{equation}

\noindent Combining \eqref{roundwinner} with the rate for round $i$ in \eqref{losers}, the condition for which the gambler should place a bet becomes:
\begin{align*}
E[\log b^*(X_i \mid X^{i-1},Y^{i-1})-&\log b^*(X_i \mid X^{i-1}) \mid x^{i-1},y^{i-1}] \\
&=\sum_{x_i}p(x_i\mid x^{i-1},y^{i-1})\log \frac{b^*(x_i \mid x^{i-1},y^{i-1})}{b^*(x_i \mid x^{i-1})} \\
&=\sum_{x_i}p(x_i\mid x^{i-1},y^{i-1})\log \frac{p(x_i \mid x^{i-1},y^{i-1})}{p(x_i \mid x^{i-1})} \\
&= \cyx(x^{i-1},y^{i-1}) \\
& > -\log \delta
\end{align*}

\noindent Thus we can see that while the DI represents the \emph{time averaged} expected increase in wealth growth rate resulting from side information, the proposed measure gives the \emph{per round} expected increase. It is important to note that with problems in communication theory, low probability events may indeed be of little concern, and thus the DI may be the correct technique with which to analyze the relationship between $Y$ and $X$. In the case of betting and the applications discussed in Section \ref{djhs}, we note that there may be great interest in how the two time series interact for specific realizations, even if those realizations are rare.

\section{Estimating the Causal Measure}\label{estimation}

An estimate of the causal measure can be obtained by simply estimating the complete and restricted distributions and then computing the KL divergence between the two at each time. Such an estimator allows us to leverage results from the field of sequential prediction \cite{merhav1998universal}. The sequential prediction problem formulation we consider is as follows: for each round $i \in \{1,\dots,n\}$, having observed some history $\history{i}$, a learner selects a probability assignment $\hat{p}_i \in \mc{P}$, where $\mc{P}$ is the space of probability distributions over $\mc{X}$. Once $\hat{p}_i$ is chosen, $x_i$ is revealed and a loss $l(\hat{p}_i,x_i)$ is incurred by the learner, where the loss function $l:\mc{X}\rightarrow \mathbb{R}$ is chosen to be the self-information loss given by $l(p,x) = -\log p(x)$.

The performance of sequential predictors may be assessed using a notion of \emph{regret} with respect to a reference class of probability distributions $\refclass \subset \mc{P}$. For a given round $i$ and reference distribution $\tilde{p}_i \in \refclass$, the learner's regret is:
\begin{equation}
r(\hat{p}_i,\tilde{p}_i,x_i) = l(\hat{p}_i,x_i) - l(\tilde{p}_i,x_i)
\end{equation}

\noindent In many cases the performance of sequential predictors will be measured by the worst case regret, given by:
\begin{align}
R_n(\refclass_n) &= \sup_{x^n \in \mc{X}^n} \sum_{i=1}^n l(\hat{p}_i,x_i) - \inf_{\tilde{p}\in \refclass_n} \sum_{i=1}^n l(\tilde{p}_i,x_i) \label{optimal_f} \\
&\triangleq \sup_{x^n \in \mc{X}^n} \sum_{i=1}^n r(\hat{p}_i,f^*_i,x_i)
\end{align}

\noindent where $p^*_i \in \refclass$ is defined as the distribution from the reference class with the smallest cumulative loss up to time $n$, i.e. the $\tilde{p}_i$ for which $R_n$ is largest. We also define $p^* \in \refclass_n \subset \mc{P}^n$ to be the cumulative loss minimizing \emph{joint} distribution, noting that the reference class of joint distributions $\refclass_n$ is not necessarily equal to $\refclass^n$ (i.e. $\refclass \times \refclass \times \dots$), as often times there may be a constraint on the selection of the best reference distribution that is imposed in order to establish bounds. In the absence of any restrictions, the reference distributions may be selected at each time such that $p^*_i(x_i)=1$, resulting in zero cumulative loss for any sequence $x^n$. Thus, sequential prediction problems impose restrictions on the reference distributions with which to compare predictor performance \cite{merhav1998universal}. For example, one may assume stationarity by enforcing $p_1^*=p_2^*=\dots=p_n^*$ or assume that $p_i^* = p^*_{i+1}$ for all but some small number of indices. For various learning algorithms (i.e. strategies for selecting $\hat{p}_i$ given $\history{i}$) and reference classes $\refclass_n$, these bounds on the worst case regret are defined as a function of the sequence length $n$:
\begin{equation}
R_n(\refclass_n) \le M(n)
\end{equation}

It follows naturally that an estimator for our causal measure can be constructed by building two sequential predictors. The restricted predictor $\estfxr{i}$ computed at each round using $\hr$, and the complete predictor $\estfxc{i}$ computed at each round using $\hc$. It then follows that each of these predictors will have an associated worst case regret, given by $R^{(r)}_n(\refclassr_n)$ and $R^{(c)}_n(\refclassc_n)$, where $\refclassr_n$ and $\refclassc_n$ represent the restricted and complete reference classes. Using these sequential predictors, we define our estimated causal influence from $Y$ to $X$ at time $i$ as:
\begin{equation}
\estcyx(i) = \kl{\estfxc{i}}{\estfxr{i}}
\end{equation}

\noindent It should be noted that when averaged over time, this estimator becomes a universal estimator of the directed information rate for certain predictors and classes of signals \cite{jiao2013universal}.

To assess the performance of an estimate of the causal measure, we define a notion of causality regret:
\begin{equation}
CR(n) \triangleq \sum_{i=1}^n \left| \estcyx(i) - \optcyx(i)  \right|
\end{equation}

\noindent where we define:
\begin{equation}
\optcyx(i) = \kl{\optfxc{i}}{\optfxr{i}}
\end{equation}

\noindent with $\optfxc{i} \in \refclass^{(c)}$ and $\optfxr{i} \in \refclass^{(r)}$ defined as the loss minimizing distributions from the complete and restricted reference classes. We note that with this notion of causal regret, the estimated causal measure is being compared against the best estimate of the causal measure from within a reference class. As such, we limit our consideration to the scenario in which the reference classes are sufficiently representative of the true sequences to produce a desirable $\optcyx$ (i.e. $\optcyx(i) \approx \cyx(i)$ for all $i$).

We now present the necessary assumptions for proving a finite sample bound on the estimates of causality regret.

\begin{assumption} \label{assumption:abscont}
For sequential predictors \estfxc{i} and \estfxr{i} and observations $(x^n,y^n,z^n)\in \mc{X}^n \times \mc{Y}^n \times \mc{Z}^n$, we assume that \estfxc{i} and \estfxr{i} are absolutely continuous with respect to each other, i.e.:
\begin{equation}
\sup_{x \in \mc{X}} \left| \log \frac{\estfxc{i}(x)}{\estfxr{i}(x)} \right| < \infty \ \ i=1,\dots,n
\end{equation}
\end{assumption}

\noindent Clearly, the above assumption will be satisfied for any sequential prediction algorithm that does not assign zero probability to any outcomes.

\begin{assumption} \label{assumption:referencekl}
For loss minimizing distributions $\optfxc{i} \in \refclass^{(c)}$ and $\optfxr{i} \in \refclass^{(r)}$, restricted sequential predictor \estfxr{i}, and observations $(x^n,y^n,z^n)\in \mc{X}^n \times \mc{Y}^n \times \mc{Z}^n$:
\begin{equation}
\sum_{i=1}^n\left|E_{\optfxc{i}}\left[
r(\estfxr{i},\optfxr{i},X_i)
\right]\right| \le M^{(r)}(n)
\end{equation}
\end{assumption}

\noindent While it is understood that the expected regret is in general bounded by worst case regret, Assumption \ref{assumption:referencekl} requires that the reference classes are sufficiently rich that the expected regret is not too large in \emph{absolute value}. This is necessary in bounding the causality regret because unlike the regret defined by \eqref{optimal_f}, $CR(n)$ \emph{increases} when the estimated distributions outperform the regret minimizing distributions.

We now present our main theoretical result, a finite sample bound on the causality regret under Assumptions \ref{assumption:abscont} and \ref{assumption:referencekl}:
\begin{theorem} \label{thm:main_result}
Let the worst case regret for the predictors $\estfxr{i}$ and $\estfxc{i}$ be bounded by $R^{(r)}_n(\refclassr_n) \le M^{(r)}(n)$ and $R^{(c)}_n(\refclassc_n) \le M^{(c)}(n)$, respectively. Then, for any collection of observations $(x^n,y^n,z^n)\in \mc{X}^n \times \mc{Y}^n \times \mc{Z}^n$ satisfying Assumptions \ref{assumption:abscont} and \ref{assumption:referencekl}, we have:
\begin{equation} \label{causal_bound}
CR(n) \le
M^{(c)}(n) + M^{(r)}(n) +
    \frac{\left|\left| \vec{c}_n \right|\right|_2}{\sqrt{2}}\sqrt{M^{(c)}(n)}.
\end{equation}
\noindent where $\vec{c}_n = [c_1,\dots,c_n]$ is a vector with elements:
\begin{equation} \label{bound_const}
c_i = \sum_{x\in \mc{X}} \left| \log \frac{\estfxc{i}(x)}{\estfxr{i}(x)} \right|
\end{equation}
\end{theorem}

\noindent A proof of the theorem may be found in Appendix \ref{app:main_result}. We note that because each $c_i$ depends solely on the estimated complete and restricted distributions, a finite sample bound may be computed at each point in time. If we make the additional assumption that the absolute log ratio of our complete and restricted predictors is bounded:
\begin{equation}\label{bound_ratio}
\sup_{x\in\mc{X}} \left| \log \frac{\estfxc{i}(x)}{\estfxr{i}(x)} \right| \le L \ \ i=1,2,\dots
\end{equation}

\noindent then we can simplify the bound by observing that:
\begin{equation}
\left|\left| \vec{c}_n \right|\right|_2 \le
L \left|\mc{X}\right| \sqrt{n}.
\end{equation}

\noindent When such a scenario holds, we can make use of the following Corollary to Theorem \ref{thm:main_result} regarding the asymptotic behavior of the causality regret:

\begin{corollary}
Let the worst case regret for the predictors $\estfxr{i}$ and $\estfxc{i}$ be sublinear in $n$ and the absolute log ratio of the complete and restricted sequential predictors be bounded as in \eqref{bound_ratio}. Then, under Assumptions \ref{assumption:abscont} and \ref{assumption:referencekl}, for any collection of observations $(x^n,y^n,z^n)\in \mc{X}^n \times \mc{Y}^n \times \mc{Z}^n$, the causality regret will be sublinear in $n$:
\begin{equation}
\lim_{n\rightarrow \infty} \frac{1}{n}CR(n) = 0
\end{equation}
\end{corollary}

Lastly, we note that in the special case where the true complete and restricted distributions are in the reference classes (i.e. $\fxr{i}\in\refclassr$ and $\fxc{i}\in\refclassc$), then under an appropriately modified Assumption \ref{assumption:referencekl} with $\fxc{i}$ and $\fxr{i}$ substituted for $\optfxc{i}$ and $\optfxr{i}$, we have that:
\begin{equation}\label{true_causal_regret}
\sum_{i=1}^n \left| \estcyx(i) - \cyx(i) \right|
\le CR(n).
\end{equation}

\noindent While in practice it is not expected that we would know whether or not the true underlying distribution is in a particular class of reference distributions, this observation will be used in performing simulations in Section \ref{simulations}.

\subsection{Addressing Infinite Order Restricted Models} \label{inf_order}

It is clear that the proposed causality regret only serves as a meaningful metric of estimation accuracy insofar as the reference class optimal causal measure $\optcyx$ serves as a useful proxy for the true causal measure \cyx. This consideration is not unique to the proposed causal measure. In an extensive analysis of problems encountered when using Granger causality, \cite{stokes2017study} describes a bias-variance tradeoff that results from the fact that subsets of VAR models will in general be of infinite order even if the complete VAR model is finite order. In the context of our estimation framework, this tradeoff lies in the selection of reference classes $\refclassr_n$ and $\refclassc_n$, which need to be rich enough to yield sufficiently good $\optcyx$ but not so rich that there do not exist sequential prediction methods for which low cumulative regret may be achieved. This issue appears in numerous forms throughout the DI literature. In particular, in order to reliably estimate the DI, it is typically assumed that in addition to all processes being jointly $d$-Markov, the subset of processes that does not include the directed information source is itself $d$-Markov \cite{jiao2013universal,murin2017k,murin2016tracking,quinn2015directed,oselio2017dynamic}. In the present context, this equates to assuming that (i) $X$, $Y$, and $Z$ are jointly $d$-Markov and (ii) $X$ is ``conditionally $d$-Markov given $Z$,'' i.e. $p(X_i\mid X^{i-1},Z^{i-1})=p(X_i\mid X_{i-d}^{i-1},Z_{i-d}^{i-1})$. While these assumptions are widely utilized to establish performance guarantees of DI estimators, their implications on the nature of the relationship between $X$ and $Y$ is absent in the literature. Thus, we seek to identify \emph{when} stationary jointly Markov processes $X$, $Y$, and $Z$ are such that $X$ is conditionally Markov given $Z$. It should be noted that, by letting $Z=\emptyset$, the following results can be extended to the standard setting consisting of only $X$ and $Y$.

To understand the conditions under which the desired independence relationships hold, we can leverage tools from Bayesian networks, which can be used to represent conditional independencies in collections of random variables using a directed acyclic graph (DAG) $\mcal{G}=(V,E)$, where $V=\{V_1,\dots,V_m\}$ is a set of random variables (equivalently nodes or vertices) and $E\subset V\times V$ is a set of directed edges that it do not contain any cycles \cite{spirtes2000causation}. The parent set of a node $V_i$ in a DAG is defined as the set of nodes with arrows going into $V_i$, $\mc{P}_i \triangleq \{V_j:(V_j \ra V_i)\in E\}$. The defining characteristic of a Bayesian network representation of a joint distribution over the nodes $V\sim p$ is the ability to factorize the distribution as:
\begin{equation}\label{factorization}
p(V) = \prod_{i=1}^m p(V_i \mid \mc{P}_i).
\end{equation}
\noindent If this factorization holds for a given $p$ and $\mcal{G}$, we say $\mcal{G}$ is a Bayesian network for $p$. A key concept when working with Bayesian networks is the d-separation criterion, which is used to identify subsets of nodes whose conditional independence is implied by the graphical structure. In particular, when given three disjoint subsets of nodes $A,B,C\subset V$ in a graph $\mcal{G}$, a straightforward algorithm (shown in Algorithm \ref{alg:dsep}) can be used to determine if $C$ d-separates $A$ and $B$. When $C$ d-separates $A$ and $B$, then for any joint distribution $p(V)$ such that $\mc{G}$ is a Bayesian network for $p$, $A$ and $B$ will be conditionally independent given $C$. While the converse is not true in general (i.e. independence does not imply d-separation), it has been shown that for specific classes of Bayesian networks, the set of parameters for which the converse does \emph{not} hold has Lebesgue measure zero \cite{spirtes2000causation,meek2013strong}. When a graph $\mcal{G}$ and joint distribution $p$ are such that d-separation holds if and only if conditional independence holds for all subsets of nodes, then the distribution $p$ is called ``faithful'' to $\mcal{G}$ \cite{spirtes2000causation}.
\begin{algorithm}[H]
\caption{d-Separation \cite{lauritzen1990independence}} \label{alg:dsep}
\hspace*{\algorithmicindent} \textbf{Input}: DAG $\mcal{G}=(V,E)$ and disjoint sets $A,B,C\subset V$
\begin{algorithmic}[1]
\State Create a subgraph containing only nodes in $A$, $B$, or $C$ or with a directed path to $A$, $B$, or $C$
\State Connect with an undirected edge any two variables that share a common child
\State For each $c\in C$, remove $c$ and any edge connected to $c$
\State Make every edge an undirected edge
\State Conclude that $A$ and $B$ are d-separated by $C$ if and only if there is no path connecting $A$ and $B$
\end{algorithmic}
\end{algorithm}

A Bayesian network is a very natural representation for collections of Markov processes. In particular, using the chain rule to factorize the joint distribution over $n$ time steps of the processes $(X,Y,Z)$ yields:
\begin{equation}
p(X^n,Y^n,Z^n)
=\prod_{i=1}^n
p(X_i,Y_i,Z_i \mid X^{i-1}_{i-d},Y^{i-1}_{i-d},Z^{i-1}_{i-d})\label{factorized}.
\end{equation}
\noindent Assuming that there are no instantaneous influences facilitates construction of a Bayesian network, as we can rely on the arrow of time to determine the direction of arrows in the network. In the presence of instantaneous influences, we cannot construct a \emph{unique} Bayesian network representation of Markov processes without making alternative assumptions. This is similar reasoning to that of \cite{quinn2015directed}, where the absence of instantaneous influences is used to establish the equivalence between DI graphs and minimal generative model graphs. Under this assumption, we can further simplify \eqref{factorized} as:
\begin{equation}
p(X^n,Y^n,Z^n)
=\prod_{i=1}^n \prod_{S\in\{X_i,Y_i,Z_i\}}
p(S \mid X^{i-1}_{i-d},Y^{i-1}_{i-d},Z^{i-1}_{i-d})\label{time_network}.
\end{equation}
\noindent Comparing \eqref{factorization} and \eqref{time_network}, it is clear that we can represent a collection of processes as a Bayesian network by letting each node be a single time point of a process (i.e. $X_i$, $Y_i$, or $Z_i$) with parents $\mc{P}_{X_i},\mc{P}_{Y_i},\mc{P}_{Z_i}\subseteq\{X^{i-1}_{i-d},Y^{i-1}_{i-d},Z^{i-1}_{i-d}\}$. Given that there may be multiple valid Bayesian networks for a particular distribution, we note that $X_i$, $Y_i$, and $Z_i$ may not be conditionally dependent on the entire set $\{X^{i-1}_{i-d},Y^{i-1}_{i-d},Z^{i-1}_{i-d}\}$. Thus, when constructing a Bayesian network for $(X,Y,Z)$ we include an edge $S_{i-k}\ra S'_{i}$ for $S,S'\in \{X,Y,Z\}$ and $k=1,\dots,d$ only if:
\begin{equation}\label{construction}
I(S_{i-k};S'_i\mid\{X^{i-1}_{i-d},Y^{i-1}_{i-d},Z^{i-1}_{i-d}\}\setminus S_{i-k})>0.
\end{equation}
Using the Bayesian network construction given by \eqref{construction}, we can leverage the d-separation criterion to gain a better understanding of the types of conditions which give rise to the conditional independence relationships needed for DI estimation. To start, we identify necessary and sufficient conditions for which $X_i$ will be d-separated from $(X^{i-l-1},Z^{i-l-1})$ by $(X^{i-1}_{i-l},Z^{i-1}_{i-l})$:
\begin{theorem} \label{thm:dsep}
Let $(X,Y,Z)$ be a collection of jointly stationary $d$-Markov processes such that there are no instantaneous influences. If $\ccdiyxz=0$, then $X$ is conditionally $d$-Markov given $Z$. If $\ccdiyxz>0$, $X$ is conditionally Markov given $Z$ of order $2d$ or less if:
\begin{equation}\label{sufficient}
I(Y_j;Y_k\mid X^i,Z^i)=0 \ \forall j \le k \le i
\end{equation}
If $\ccdiyxz>0$ but \eqref{sufficient} is not satisfied, there will not exist any positive integer $l$ such that $(X_{i-l}^{i-1},Z_{i-l}^{i-1})$ d-separates $X_i$ from $(X^{i-l-1},Z^{i-l-1})$ in the Bayesian network generated according to \eqref{construction}.
\end{theorem}

\noindent A proof of the theorem can be found in Appendix \ref{app:dsepproof}. The implication of this theorem is that the desired d-separation criteria only occurs when no two time points of $Y$ directly influence each other and $X_i$ is only causally influenced by a single $Y_j$ for some $j \le i$. In particular we note that this excludes jointly stationary $d$-Markov processes aside from the special case where $p(Y_i \mid X_{i-d}^{i-1}, Y_{i-d}^{i-1}) = p(Y_i \mid X_{i-d}^{i-1})$ and $p(X_i \mid X_{i-d}^{i-1}, Y_{i-d}^{i-1}) = p(X_i \mid X_{i-d}^{i-1}, Y_{i-\tau})$ for some $0\le \tau\le d$.

Theorem \ref{thm:dsep} uses d-separation to provide us with a characterization of networks of processes that are guaranteed to have the conditional independence relations required by DI estimators. With regard to the processes for which we cannot demonstrate d-separation (i.e. those not satisfying \eqref{sufficient}), the only distributions that will have the desired conditional independence relations are those that are \emph{unfaithful} to their graphs. While there is ample discussion in the literature noting that these distributions are typically not seen in practice (see \cite{spirtes2000causation} and citations therein), a formal characterization within the present context is desired.

\begin{theorem}\label{thm:necessary}
The set of parameters defining a collection $(X,Y,Z)$ of jointly stationary irreducible aperiodic Markov processes such that there exists a positive integer $l$ where $X$ is conditionally $l$-Markov given $Z$ but $(X_{i-l}^{i-1},Z_{i-l}^{i-1})$ does not d-separate $X_i$ from $(X^{i-l-1},Z^{i-l-1})$ in the Bayesian network constructed by \eqref{construction} has Lebesgue measure zero with respect to $\R^N$.
\end{theorem}

\noindent A proof of the theorem is given in Appendix \ref{app:zeromeasproof}.

This observation is not to detract from the literature on the estimation and application of DI, but rather to extend the concerns expressed in \cite{stokes2017study} to the information theoretic analogues of Granger causality. The issue arises not from the estimators themselves, but rather from the properties of d-separation in causal models. When $X$ is not marginally Markov, any plug-in estimator that requires an estimate of $\fxr{i}$ will necessitate approximation of an infinite order model and therefore introduce a bias-variance tradeoff.

The most relevant previous work in this area was in \cite{kontoyiannis2016estimating}, where the performance of plug-in estimators of DI were analyzed. It was noted in this work that in scenarios where $X$ is not marginally Markov, the finite order approximation is nevertheless interesting, in particular because it serves as an \emph{upper bound} for the DI. Moreover, it is shown that this approximate measure may be reliably estimated even in scenarios where $X$ is not marginally Markov.

We now propose a second alternative to DI for which reliable estimates may be obtained when $X$ is not marginally Markov. The proposed measure is complementary to that proposed in \cite{kontoyiannis2016estimating} in that it serves as a \emph{lower bound} for the DI. In order to avoid dependence on an infinite past, we consider conditioning on a ``stale'' history:
\begin{theorem}\label{thm:markov2}
Let $(X,Y)\sim p$ be a jointly stationary irreducible aperiodic finite-alphabet $d$-Markov process. For a fixed $k$, define $\tilde{X}_i \triangleq (X_i,Y_{i-k+1})$. Then $\tilde{X}$ is a jointly stationary irreducible aperiodic $(d+k)$-Markov process and the following equality holds:
\begin{equation} \label{xtildemarkov}
p(X_i\mid X^{i-1},Y^{i-k})=p(X_i\mid X_{i-k-d}^{i-1},Y_{i-k-d}^{i-k}).
\end{equation}
\end{theorem}

\noindent The above theorem states that so long as the distribution of $X_i$ is conditioned upon its own past and any $d$ consecutive samples of $Y$, then it is independent of all $X$ and $Y$ that precede those samples of $Y$. The proof of the theorem may be found in Appendix \ref{app:thms}.

It follows that by conditioning on some stale past of $Y$, a proxy for DI may obtained. Formally, define the partial history with lag (or staleness) $k$ to be:
\begin{equation}
\hp{k} \triangleq \hr \cup \{y_1,\dots,y_{i-k}\}.
\end{equation}

\noindent Similarly, define the partial conditional distribution:
\begin{equation}
\fxp{i}{k}(x_i) \triangleq p(x_i \mid \hp{k}).
\end{equation}

\noindent We note that the partial conditional distribution is a generalization of the complete and restricted distributions in that $\fxp{i}{1}=\fxc{i}$ and $\fxp{i}{i}=\fxr{i}$. Finally, we define the partial causal measure with lag $k$ to be:
\begin{equation}\label{partialc}
\cyxp{k}(\hc) \triangleq
\kl{\fxc{i}}{\fxp{i}{k}}
\end{equation}

\noindent This sample path dependent measure may be related to the model dependent measures by defining a partial DI (PDI) and PDI rate:
\begin{align}
I(Y_{n-k}^{n-1}\rightarrow X^n) &\triangleq
\sum_{i=1}^n E\left[ \cyxp{k}(\hc) \right] \label{partialdi}\\
\bar{I}_P^{(k)}(Y\rightarrow X) &\triangleq \lim_{n\rightarrow\infty}\frac{1}{n}
\sum_{i=1}^n I(Y_{n-k}^{n-1}\rightarrow X^n) \label{partialdirate}
\end{align}

\noindent Much like the DI (rate), the PDI (rate) can be represented as a difference of entropies (rates):
\begin{align}
I(Y_{n-k}^{n-1}\rightarrow X^n) &= H(X^n\mid\mid Y^{n-k-1}) - H(X^n\mid\mid Y^{n-1}) \label{partialdient}\\
\bar{I}_P^{(k)}(Y\rightarrow X) &= \bar{H}^{(k-1)}(X\mid\mid Y) - \bar{H}^{(1)}(X\mid\mid Y)\label{partialdirateent}
\end{align}

\noindent where we have replaced the first entropy terms on the right hand side of \eqref{dimentr} and \eqref{dimrateentr} with a lagged causally conditioned entropy.

The partial causal measures of \eqref{partialc}, \eqref{partialdi}, and \eqref{partialdirate} have straightforward interpretations in relationship to their complete counterparts. In particular, we note that these measures can be viewed as measuring the causal effect of the \emph{recent} past of $Y$ on $X$. Additionally, the PDI serves as a lower bound for DI. Therefore, effectively estimating the PDI for scenarios in which we do not have a universal estimator of the DI may be of great interest. The estimators of \cite{jiao2013universal} can be extended to be universal estimators of the PDI rate:

\begin{theorem}\label{thm:pdi_est}
Let $(X,Y)\sim p$ be a jointly stationary irreducible aperiodic finite-alphabet Markov process of order $d$ or less. Let $\estfxc{i}$ be a depth-$d$ CTW estimate of $\fxc{i}$ with access to $\hc$ and $\estfxk{i}$ be a depth-$(d+k)$ CTW estimate of $\fxp{i}{k}$ with access to $\hp{k}$. Then:
\begin{equation}
\lim_{n\rightarrow \infty} \frac{1}{n}
\sum_{i=1}^n \kl{\estfxc{i}}{\estfxk{i}} =
\bar{I}_P^{(k)}(Y\rightarrow X)
\ \ p_{X^n,Y^n}-a.s.
\end{equation}
\end{theorem}

\noindent The proof of the theorem may be found in Appendix \ref{app:thms}. We note that this theorem is analogous with Theorem 3 from \cite{jiao2013universal}, where here we have removed any assumptions about how $X$ behaves marginally. It should also be noted that here we are only considering one of four estimators proposed by \cite{jiao2013universal}. Presumably, similar results could be obtained for the other estimators, though this is not the primary concern of this work.

As a final point, we note that by combining the PDI and the work of \cite{kontoyiannis2016estimating}, we can reliably estimate upper an lower bounds of DI for jointly Markov processes regardless of whether or not the sub-processes are themselves marginally Markov:
\begin{proposition}\label{prop:di_bounds}
Let $(X,Y)$ be a jointly stationary Markov process of order $d$.
Define the truncated DI (TDI) rate of order $k$ as:
\begin{equation}\label{tildedi}
\bar{I}_T^{(k)}(Y\rightarrow X) \triangleq
\lim_{n\rightarrow \infty} \frac{1}{n} \sum_{i=1}^n
I(X_i;Y^{i-1}_{i-k}\mid X_{i-k}^{i-1}).
\end{equation}
Then for $k_1 \ge 1$ and $k_2 \ge d$, we have:
\begin{equation}
\bar{I}_P^{(k_1)}(Y\rightarrow X)
\le
\dirateyx
\le
\bar{I}_T^{(k_2)}(Y\rightarrow X)
\end{equation}
\noindent with both bounds becoming equalities as $k_1$ and $k_2$ approach infinity.
\end{proposition}
\noindent The proof of the proposition is given in Appendix \ref{app:di_bounds}. It should be noted that, in a scenario where both joint and marginal Markovicity hold, the TDI and DI rates are equivalent. We have here simply renamed the TDI to emphasize that, in many cases, they are in fact \emph{not} equivalent. The implication of this proposition is that, for the large class of jointly Markov processes for which the sub-processes are not marginally Markov (as characterized by Theorems \ref{thm:dsep} and \ref{thm:necessary}), we can avoid the bias-variance tradeoff resulting from using truncated approximations of infinite-order models by reliably estimating \emph{bounds} on DI rather than the DI itself. In particular, the quantity that is estimated under the assumption that $X$ is marginally Markov (given by \eqref{tildedi}) provides a lower bound, and the PDI resulting from conditioning on a stale history provides an upper bound. Furthermore, both of these quantities can be reliably estimated, as demonstrated by \cite{kontoyiannis2016estimating} and Theorem \ref{thm:pdi_est}.

\section{Results}\label{applications}

\subsection{Stationary Markov Processes}\label{simulations}

We begin by demonstrating estimation of the causal measure for simulated stationary first order Markov processes using a context tree weighting (CTW) sequential prediction algorithm \cite{willems1995context}. For the purpose of estimating either the complete conditional distribution $\fxc{i}$, or partial conditional distribution $\fxp{i}{k}$, we utilize a CTW with side information as in \cite{cai2005universal}. In order to evaluate the causality regret of a given estimator of the causal measure, we need worst case regret bounds on the sequential predictors utilized in the estimator.

\begin{lemma}[\cite{tjalkens1993sequential,jiao2013universal}]
Let $\hat{p}$ be a depth-$d$ CTW probability assignment of a stationary finite-alphabet Markov process $X\sim p$ of order less than or equal to $d$. Then the worst case regret is bounded as:
\begin{equation}
\sup_{x^n} \ \log \frac{p(x^n)}{\hat{p}(x^n)} \le
\frac{(\left| \mc{X} \right|-1)L}{2}
\log\frac{n}{L}
+L\left(
\frac{\left|\mc{X}\right|}{\left|\mc{X}\right|-1} + \log \left|\mc{X}\right|
\right) - \frac{1}{\left|\mc{X}\right|-1} \label{ctw_regret}
\end{equation}
where $L$ is the number of leaves in the CTW (equivalently, the number of states of $X$).
\end{lemma}

\noindent Next we note that this bound may be extended to the case where $X$ is given access to causal side information $Y$:

\begin{proposition}\label{prop:sideinfo_regret}
Let $\hat{p}$ be a depth-$d$ CTW probability assignment of $X$ causally conditioned on $Y$, where $(X,Y) \sim p$ is a pair of jointly stationary finite-alphabet process of order less than or equal to $d$. Then the worst case regret is bounded as:
\begin{equation}
\sup_{x^n,y^n} \ \log \frac{p(x^n\mid\mid y^n)}{\hat{p}(x^n\mid\mid y^n)} \le
\frac{(\left| \mc{X} \right|-1)L}{2}
\log\frac{n}{L}+L\left(\left|\mc{X}\right|-1\right) + S \label{ctw_sideinfo_regret}
\end{equation}
where $L$ is the number of leaves in the CTW, $S$ is the total number of nodes in the CTW.
\end{proposition}

\noindent A proof of the proposition is provided in Appendix \ref{app:sideinfo_regret}. Using the above lemma and proposition we can compute the values of the causality regret bounds by using \eqref{ctw_regret} and \eqref{ctw_sideinfo_regret} for $M^{(r)}(n)$ and $M^{(c)}(n)$, respectively, in \eqref{causal_bound}. In the following sections we compare the causal regret bound with the true estimation accuracy for three scenarios.
\subsubsection{Independent Processes}
Let $X$ and $Y$ be independent ternary processes (i.e. $X_i,Y_i \in \{0,1,2\}$ for all $i$), with each process being stationary first order Markov. As such, the processes are completely characterized by the probabilities $p(x_i \mid x_{i-1})$ and $p(y_i \mid y_{i-1})$ for $x_i$, $x_{i-1}$, $y_i$, $y_{i-1} \in \{0,1,2\}$. Given the independence of $X$ and $Y$, we have that for all $i=1,2,\dots$, $\cyx(i)=\cxy(i)=0$.

Figure \ref{fig:indep} shows the estimate of the causal measure over time for $n=10000$ samples. We can see that the estimated causal measure in both directions quickly becomes very small at all times. The true causal measure is not shown because it is always zero. In the bottom panel of the figure we see the normalized causal regret with respect to the true causal measure as in \eqref{true_causal_regret}, which in this case is given by the running average of the estimated causal measure. Additionally, we show the causal regret bounds, which are computed using $\left|\mc{X}\right|=3$, $L=3$ in \eqref{ctw_regret}, and $L=9$ and $S=10$ in \eqref{ctw_sideinfo_regret}.
\begin{figure}[ht!]
  \centering\includegraphics[keepaspectratio,width=0.6\textwidth]{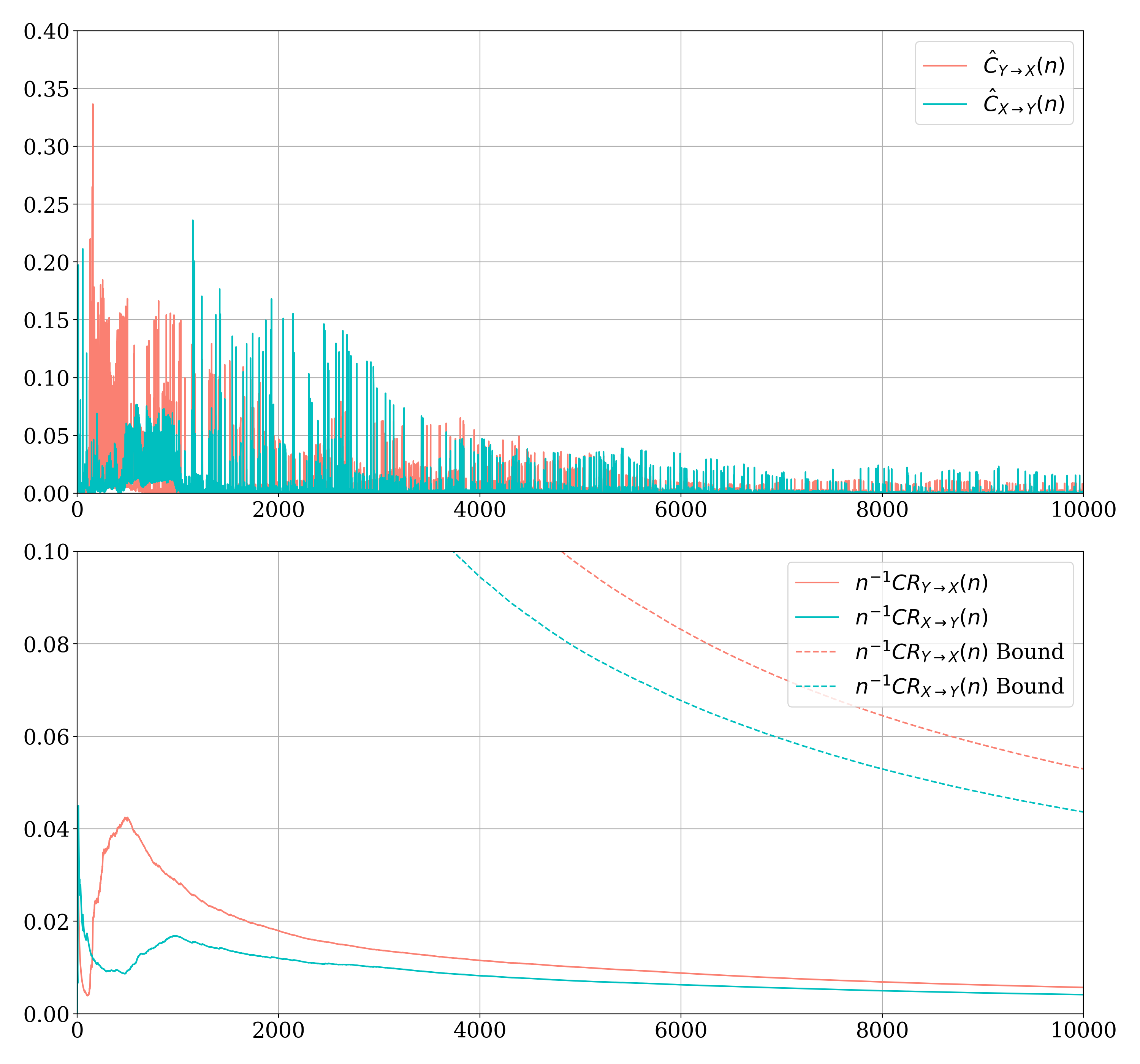}
  \caption{Top - Estimates of causal measure in each direction for independent processes. Bottom - Normalized cumulative absolute error of estimates (solid) and normalized causality regret bounds (dashed).}
  \label{fig:indep}
\end{figure}

\subsubsection{Unidirectional IID Influence}
For the second scenario we consider a pair of processes wherein each $Y_i$ is independent and identically distributed and $X_i$ is dependent only upon $X_{i-1}$ and $Y_{i-1}$. As such, it is clear that, in addition to $X$ and $Y$ being jointly first-order Markov, $X$ is marginally first-order Markov. While the marginal Markovicity is immediately clear in this case, we point out that these processes do indeed satisfy the condition of Theorem \ref{thm:dsep} for any parameterization of the probabilities in question.

Figure \ref{fig:unidir} shows the true causal measure $\cyx$ alongside the estimates $\estcyx$ and $\estcxy$. For clarity, only the last 100 time points are shown. We can see that in this time window the estimated $\estcyx$ tracks the true causal measure $\cyx$ very well, and the estimated $\estcxy$ has converged to 0 as desired. In the bottom panel we see that, because the causal measure $\estcxy(i)<\estcyx(i)$, $c_i$ in equation \eqref{bound_const} is much smaller in the $X\rightarrow Y$ direction and thus the causal regret bound is considerably tighter. This is consistent with the result obtained in \cite{kontoyiannis2016estimating} that plug-in estimators of the DI rate will converge at a faster rate if the DI rate is zero.

\begin{figure}[ht!]
  \centering\includegraphics[keepaspectratio, width=0.6\textwidth]{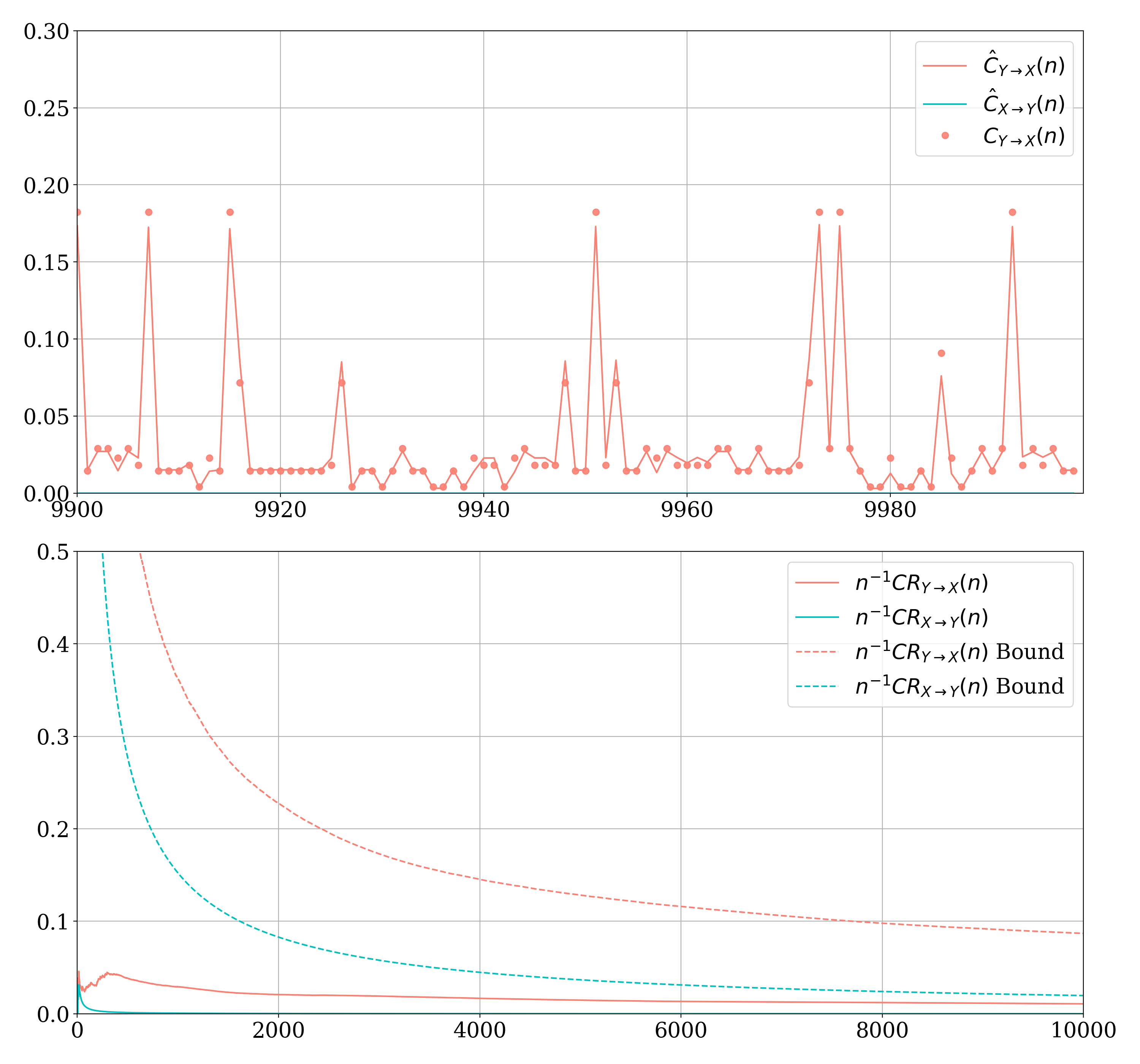}
  \caption{Top - Estimates of causal measure in each direction for unidirectional influences. Bottom - Normalized cumulative absolute error of estimate (solid) and normalized causality regret bounds (dashed).}
  \label{fig:unidir}
\end{figure}

\subsubsection{Bidirectional Influences}
Lastly, we consider the scenario where $X$ and $Y$ mutually influence each other. Specifically, let each $X_i$ and $Y_i$ be independently influenced by $X_{i-1}$ and $Y_{i-1}$ such that the processes are fully characterized by the probabilities $p(x_i \mid x_{i-1},y_{i-1})$ and $p(y_i \mid x_{i-1},y_{i-1})$.

Figure \ref{fig:bidir} shows the true and estimated causal measures in both directions. The bottom panel shows the cumulative absolute error alongside the causal regret bounds. We note that here we have extended the time horizon to $n=50000$ to illustrate that the estimators exhibit bias resulting from the fact that $X$ is not marginally Markov. As a result, it is important to note that the true restricted distribution $\fxr{i}$ will not be in the reference class of restricted distributions $\refclassr$ and we can expect the causality regret bound to be lower than the cumulative absolute error as $n\rightarrow\infty$. Due to the non-Markovicity of $X$, computing the true restricted distribution at each time becomes increasingly challenging. To address this, we derive a recursive updating algorithm for efficiently computing the true causal measure $\cyx(i)$ such settings. Details can be found in Appendix \ref{app:hmm}.

\begin{figure}[ht!]
  \centering\includegraphics[keepaspectratio, width=0.6\textwidth]{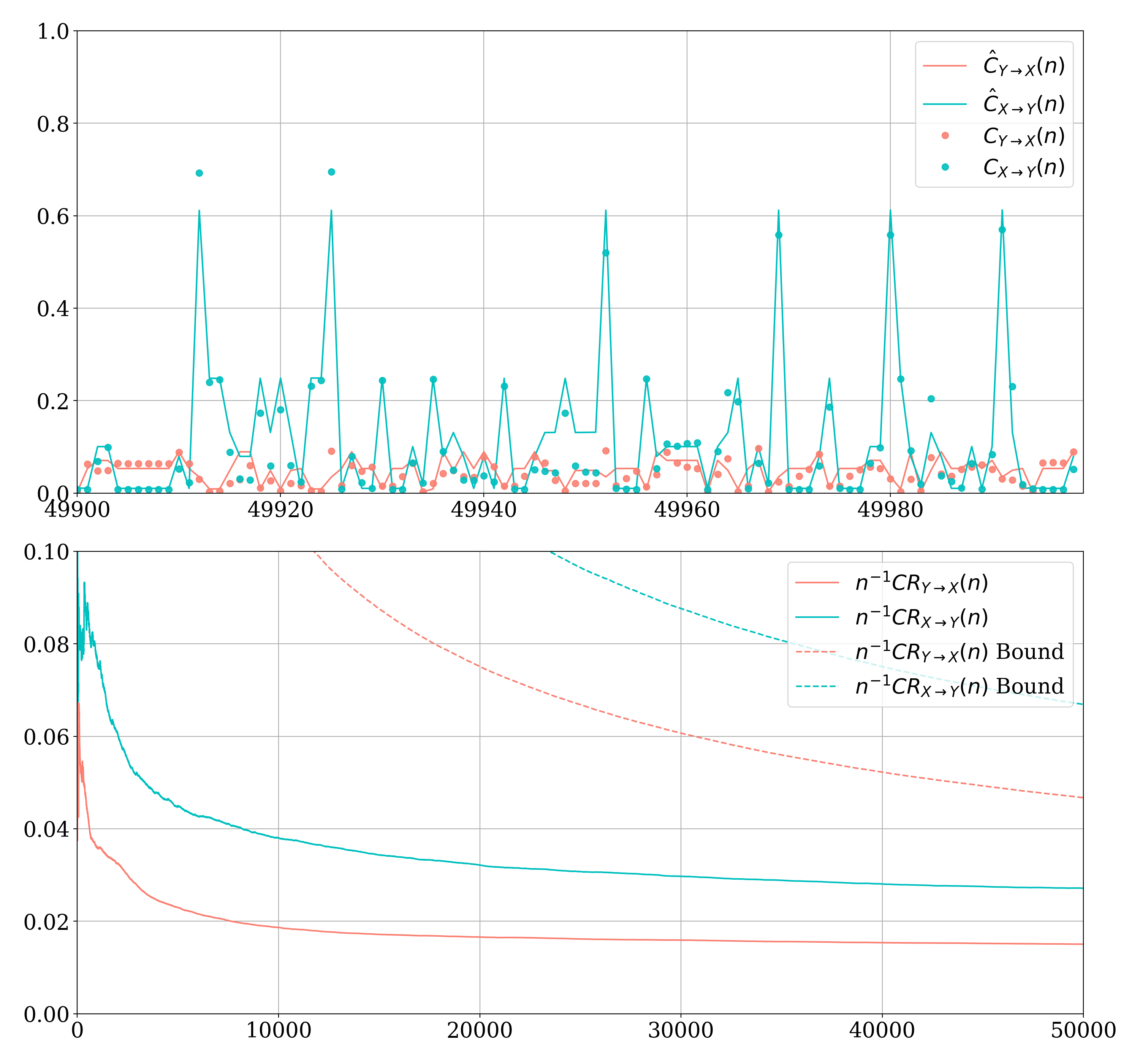}
  \caption{Top - Estimates of causal measure in each direction for bidirectional influences. Bottom - Normalized cumulative absolute error of estimate (solid) and normalized causality regret bounds (dashed).}
  \label{fig:bidir}
\end{figure}

To address the estimation bias seen in Figure \ref{fig:bidir}, we consider the partial causal measure $\cyxp{k}(i)$ defined by \eqref{partialc}. Figure \ref{fig:bidir_p} shows an estimate of the partial causal measure on the same sequence considered in Figure \ref{fig:bidir} with a staleness of $k=1$. The bottom panel of Figure \ref{fig:bidir_p} depicts the cumulative absolute error and the causal regret bounds. While the worst case regret for the complete predictor $M^{(c)}(n)$ remains the same as in the previous examples, the regret of the partial predictor is computed using equation \eqref{ctw_sideinfo_regret} with $L=27$ (3 values for $x_{i-1}$ times 9 possible values for $(x_{i-2},y_{i-2})$) and $N=31$ ($27$ leaf nodes, $3$ depth-$1$ nodes, and $1$ root node).

We can see in Figure \ref{fig:bidir_p} that due to the increased number of nodes in the CTW estimate of the partial distribution, the normalized absolute error decreases more slowly at the beginning. Regardless, the estimate of the partial causal measure does not exhibit the same behavior of converging on a biased estimate. We see the error continues to decrease throughout the entire sequence. Moreover, a visual comparison of the true and estimated measures makes clear that the estimate is unbiased.

Given that the same sequences were used in generating Figures \ref{fig:bidir} and \ref{fig:bidir_p}, we can compare the values of the complete causal measure with the partial causal measure. It is clear that while there is considerable agreement on the positions of the spikes in causal influence, the strengths vary. While it is true that the partial causal measure will be smaller than the complete causal measure in expectation (i.e. partial DI is less than DI), there are times where the stale history $y^{i-k}$ is misleading about the recent history $y_{i-k+1}^{i-1}$, and thus we sometimes see that the partial causal measure is larger than the complete causal measure on a given sample path. Lastly, we note that because the true partial distribution is in the class of reference partial distributions, the causality regret bounds will bound the cumulative absolute error.

\begin{figure}[ht!]
  \centering\includegraphics[keepaspectratio, width=0.6\textwidth]{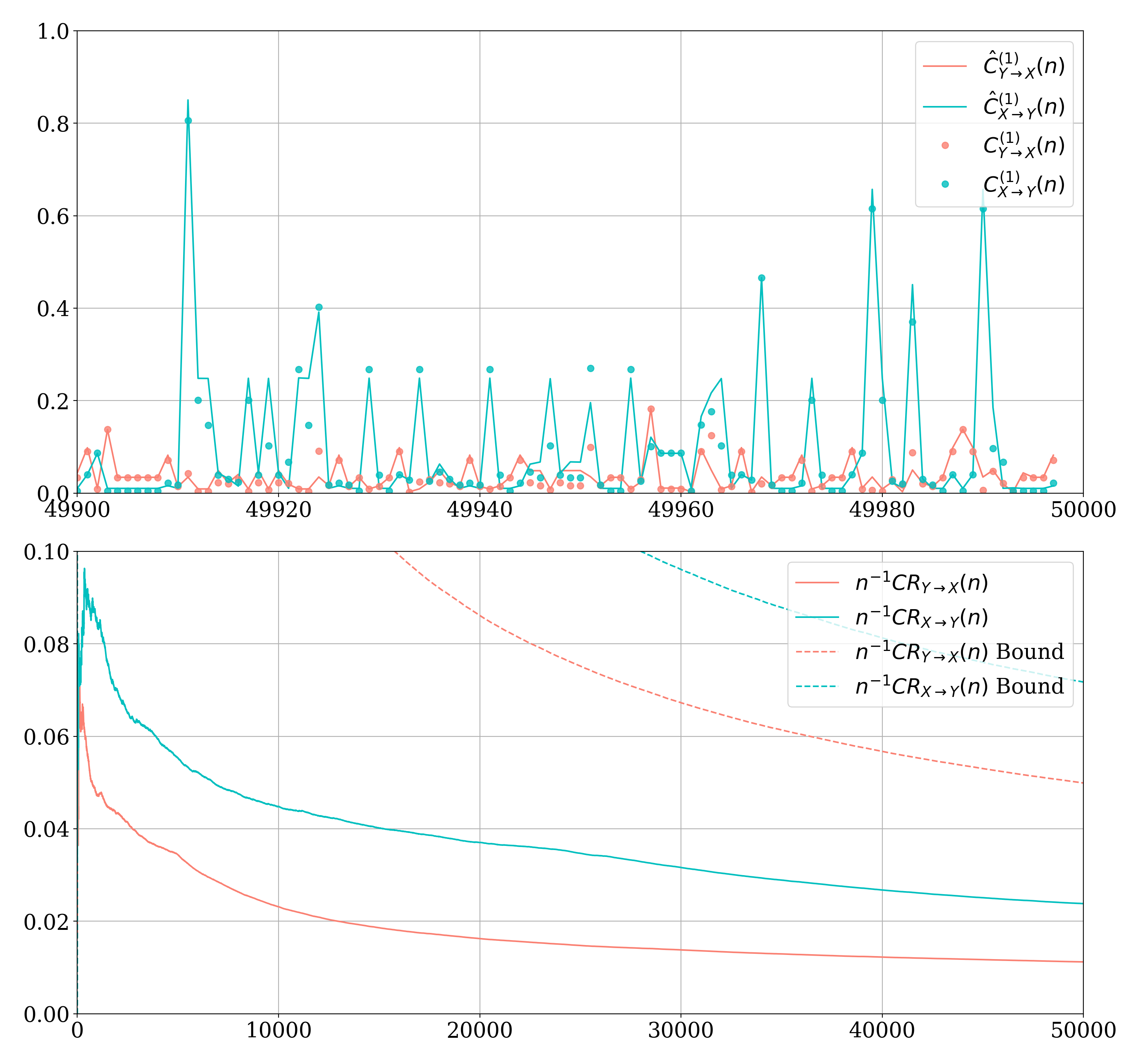}
  \caption{Top - Estimates of the partial causal measure with $k=1$ in each direction for bidirectional influences. Bottom - Normalized cumulative absolute error of estimates (solid) and normalized causality regret bounds (dashed).}
  \label{fig:bidir_p}
\end{figure}
\subsection{Stock Market Indices}\label{djhs}

\begin{figure}[ht!]
  \centering\includegraphics[keepaspectratio, width=0.7\textwidth]{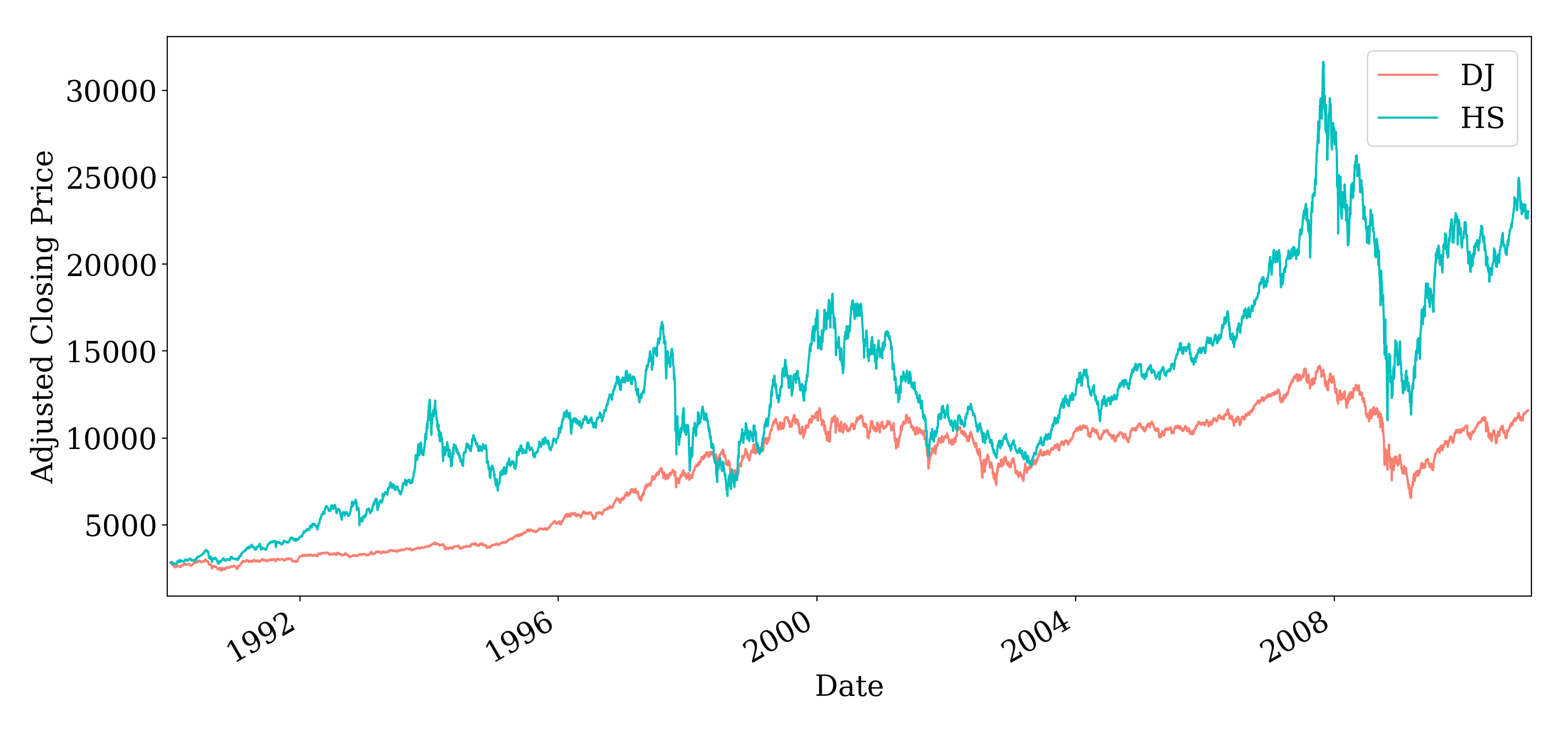}
  \caption{The Dow Jones (DJ) Industrial Average  and Hang Seng (HS) indices.}
  \label{fig:dji_hsi_t}
\end{figure}

\begin{figure*}[t!]
  \centering
  \includegraphics[keepaspectratio, width=\textwidth]{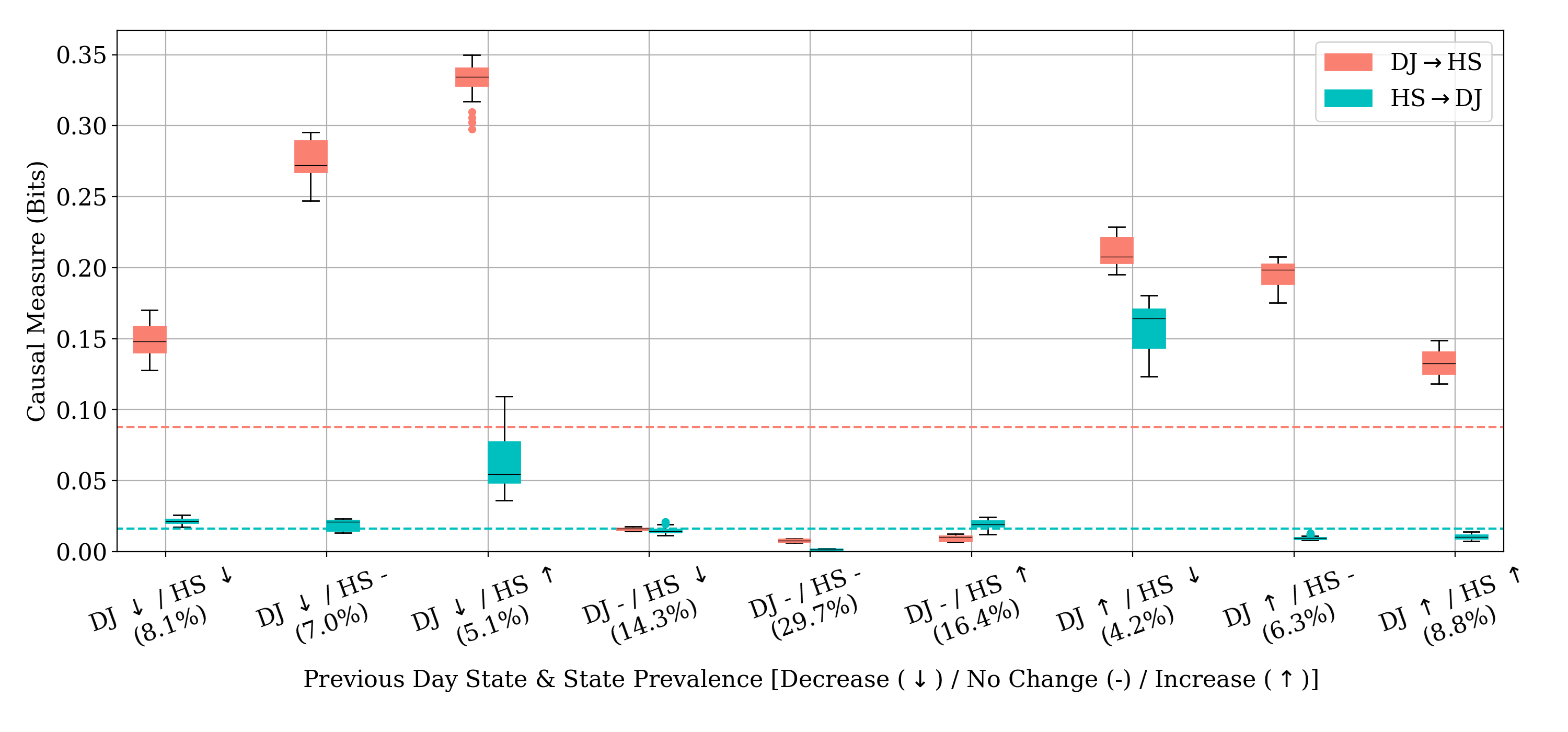}
  \caption{Causal measure between stock indices for different previous day states from 2008 to 2011. Below each of the 9 possible previous day states we include the percentage of days in which that state occurs. The dashed lines represent the DI estimate.}
  \label{fig:dji_hsi_c}
\end{figure*}

We now demonstrate the use of the sample path causal measure on historical stock market data from the Dow Jones (DJ) Industrial Average index on the New York Stock Exchange (NYSE) and the Hang Seng (HS) index on the Shanghai Stock Exchange (SSE), as in \cite{jiao2013universal}. In \cite{jiao2013universal} it was shown that the DJ index had a greater influence on the HS index than vice versa by measuring the DI between the sequences of daily changes in adjusted closing price. Here we consider the same dataset, shown in Figure \ref{fig:dji_hsi_t}.

The data was downloaded from Yahoo Finance. Given that the NYSE and SSE are closed on different holidays, missing values were interpolated on days where one was open and the other was not (weekends and shared holidays were not interpolated). We next consider the inter-day percentage change in adjusted closing price and quantize it to a ternary sequence with a value of 0 indicating a drop by more than 0.8\%, 2 indicating a rise by more than 0.8\%, and 1  representing no significant change. In computing the influence of HS on DJ, the HS data is shifted forward by one day. This is due to the fact that on each day, the NYSE closes before the SSE opens, as noted in \cite{jiao2013universal}. Thus, the HS is affected by the same day DJ, while the DJ is affected by the previous day HS.

Figure \ref{fig:dji_hsi_c} shows the estimate of the causal measure for different previous day states in each direction using depth-1 CTW predictors in addition to the estimated DI (dashed line). Values from the final 3 years of the data are shown in the box plot, with the fairly tight error bars showing that the CTW estimators have mostly converged (we would not expect complete convergence due to the non-stationarity of stock market data).

There are numerous noteworthy points in Figure \ref{fig:dji_hsi_c}. First we note that, as a result of averaging, the DI is never \emph{equal} to the causal measure in the DJ to HS direction. In particular we note that when DJ does not change, it has virtually no effect on the distribution of HS. Furthermore, most of the time DJ does not change. On the other hand, on the rare occasion that DJ went down and HS went up on the previous day (5.1\% of days), the causal measure is almost 4X the DI. Similarly, on the 4.2\% of days where DJ goes up and HS goes down, the causal measure from HS to DJ is roughly 10X the DI. When considering this type of data, the added value of Q3 over Q1 becomes very clear. If one has access to what has already happened (i.e. the previous day state), then why take an expectation over the past?

\begin{remark}\label{remark:djhs}
It is crucially important to make clear the notion of causality that is considered in this context. There is no doubt that there are confounding factors (i.e. factors that affect both DJ and HS) that would decrease the measured influence if included in the model. That having been said, it is clear that there is information contained in the DJ that provides us with an improved ability to predict the next day's HS, a finding which may certainly be of use for applications outside of classical ``causal inference''. In order to make claims of causal influence, careful attention needs to be paid to potential causes that are not considered in the model, thus requiring domain expertise. As such, in scenarios where one cannot confidently rule out the potential presence of confounding factors, the proposed measure may be more accurately viewed as a measure of increased predictability (as in \cite{kleeman2002measuring}).
\end{remark}

\section{Discussion}\label{discussion}

The concepts presented in this paper can be distilled to three primary contributions. First, we have introduced a need for measuring causal influences between random processes that depend on the sample paths of those processes. We have shown that in both simple thought experiments and real stock market data, there exist sample path dependent causal influences that may occur infrequently, and are thus not captured by average measures such as GC and DI. Second, we have proposed a measure for identifying these influences. We have shown that this measure gives results consistent with intuitions in a number of examples. Furthermore, we proved finite sample bounds on the performance of an estimator of our proposed measure. Third, we have presented a characterization of when it is possible to reliably estimate both the proposed measure and the DI. This characterization utilizes the connections between DI, GC, and causal graphical models to demonstrate that in many common cases, we should not expect to be able to get unbiased estimates of DI, which to our knowledge has been repeatedly overlooked throughout the information theory literature. To address this issue we introduce a notion of partial DI and a corresponding sample path dependent partial causal measure.

There are numerous directions for continued research in this area. Further leveraging the tools from causal graphical models can enable a better understanding of the circumstances in which we can estimate measures of causal influence reliably. Furthermore, the tools from this field are necessary to distinguish between true causal influences and measures of improved predictability. Additionally, extending the three questions proposed in the introduction to the case of general causal graphs is of great interest. In particular, how can we measure how different realizations of groups of random variables variables affect another group of random variables in a given directed graph? We believe the philosophy presented in this paper may be used to address this question.

It is important to note the present work is built upon the restrictive assumption that all processes are observed. There has been considerable recent interest in estimating causal influences and graph structures when only a subset of processes may be observed \cite{geiger2015causal,etesami2016learning,matta2018consistent}. As such, there is an opportunity to study how these results may be applied to inferring dynamic causal influences that are dependent upon \emph{realizations} of a subset of processes.

Another line of future work is further investigation of the significance of partial directed information developed in Section \ref{inf_order} and its application in quantifying information leakage for coupled systems with delayed information \cite{gorantla2012characterizing}, providing fundamental performance limitations of closed-loop systems \cite{martins2008feedback} subject to delay constraints, or in characterizing rate-performance tradeoffs \cite{tanaka2018lqg} for network control problems with non-classical information structures \cite{gorantla2011information,kulkarni2015optimizer} pertaining to information and delay constraints.

A final area for future work is the demonstration of how the causal measure can provide added value in decision making. A promising avenue lies in the use of the causal measure for aiding in \emph{time-varying model selection}. Take, for example, the stock market example in Section \ref{djhs}. It is shown in \cite{etesami2017econometric} that using DI for model selection can yield improvements in the systemic risk. A natural extension of this would be to use the sample path causal measure to create a collection of models that are dependent upon the current ``state'' of the stock market. This would enable minimizing the number of estimated parameters while ensuring that opportunities for leveraging directed influences are not overlooked.

\appendices

\section{Equivalence of Interventional and Non-Interventional Causal Measure in Section \ref{sec:examples}}\label{app:backdoor}

First, we introduce the \emph{causal model} defined by Pearl \cite[Definition 2.2.2]{pearl2009causality}, which consists of a causal structure (i.e. a DAG) and a set of functions defining a probability distribution over each node in the DAG. For the three examples in section \ref{sec:examples}, we have that the causal structure is given by Figure \ref{fig:examplegraphs}. For the first two examples, the functions are given by equations \eqref{ex_one} and \eqref{ex_two}, respectively. For the third example of horse betting, we assume that for each $i$, $X_i=f_i(X^{i-1},Y_{i-1},U_i)$ and $Y_i=V_i$, where $f_i$ is some collection of functions, $U$ is a collection of iid random variables (independent of $X$ and $Y$), and $V$ is a collection of iid random variables (independent of $X$, $Y$, and $U$). The key element of this assumption is that the winner of the $i^{\text{th}}$ race, $X_i$, is functionally dependent on the side information $Y_{i-1}$, meaning that changing the side information could change the winner. Without this technicality, the example would not constitute a causal model in the sense of \cite[Definition 2.2.2]{pearl2009causality}. Once these causal models are established, showing the equivalence between the interventional and non-interventional measures discussed in Remark \ref{remark:intervention} can easily be shown using the second rule of the so-called $do$-calculus \cite[Theorem 3]{pearl1995causal}. Specifically, showing that $p(x_i\mid x^{i-1},do(y_{i-1}))=p(x_i\mid x^{i-1},y_{i-1})$ amounts to showing that $X_i$ and $Y_{i-1}$ are d-separated by $X^{i-1}$ in an augmented DAG where the outgoing arrows from $Y_{i-1}$ have been removed. This holds trivially in all three DAGs in Figure \ref{fig:examplegraphs} because removing the outgoing arrows from $Y_{i-1}$ results in there being \emph{no} path connecting $Y_{i-1}$ to $X_i$ in the augmented DAG.

\section{Computing True Causal Measure with Hidden Markov Models}\label{app:hmm}

In order to compute the true causal measure, it is necessary to compute the true restricted distribution. As discussed in Section \ref{inf_order}, the restricted distribution is, in general, non-Markov. As such, it is desirable to have an efficient method for computing the true restricted distribution $p(y_i \mid y^{i-1})$. Here we derive update equations for recursively computing $p(y_i \mid y^{i-1})$. The proposed updating scheme is a generalization of the well known recursive method for evaluating the likelihood of a process under a standard hidden Markov model  where the likelihood is given by $p(y_i \mid x_i)$ and the one-step prediction distribution is given by $p(x_i \mid x_{i-1})$ \cite[Ch. 9]{jurafsky2014speech}.

First, assume $X$ and $Y$ are jointly first order Markov as in Section \ref{simulations} and decompose the restricted distribution as the product of ``likelihood'' and ``prior'' terms:
\begin{align*}
p(y_{i+1} \mid y^i)
&=\sum_{x_{i}} p(y_{i+1}, x_{i} \mid y^i)  \\
&=\sum_{x_{i}}
\underbrace{p(y_{i+1} \mid x_{i},y_i)}_{\text{Likelihood}}
\underbrace{p(x_{i} \mid y^i)}_{\text{Prior}}
\end{align*}

\noindent where we note that only the prior term has a long-term dependence on the past. The prior may be further decomposed into the sum of products of ``one-step prediction'' and ``posterior'' terms:
\begin{align*}
p(x_i \mid y^i)
&= \sum_{x_{i-1}} p(x_i, x_{i-1} \mid y^i) \\
&= \sum_{x_{i-1}}
\underbrace{p(x_i \mid  x_{i-1}, y_{i-1})}_{\text{One-Step Prediction}}
\underbrace{p(x_{i-1} \mid y^i)}_{\text{Posterior}}\\
\end{align*}
\noindent where now only the posterior has a long-term dependence on the past. Lastly, we can use Bayes' Rule to show that the posterior depends only on the previous likelihood evaluated at the newly observed $y_i$ and the previous prior:
\begin{align*}
p(x_{i-1} \mid y^i)
&= \frac
{p(y_i \mid x_{i-1},y^{i-1})p(x_{i-1}\mid y^{i-1})}
{\sum_{\tilde{x}_{i-1}}(y_i \mid \tilde{x}_{i-1},y^{i-1})p(\tilde{x}_{i-1}\mid y^{t-1})} \\
&= \frac
{p(y_i \mid x_{i-1},y_{i-1})p(x_{i-1}\mid y^{i-1})}
{\sum_{\tilde{x}_{i-1}}(y_i \mid \tilde{x}_{i-1},y_{i-1})p(\tilde{x}_{i-1}\mid y^{i-1})}
\end{align*}

\noindent Thus, the restricted distribution can be computed in a recursive manner. To initialize the algorithm, define $y_0=x_0=\emptyset$, $p(\emptyset\mid\cdot)=1$, and starting distributions $p(\cdot\mid\emptyset)=p(\cdot)$.
\section{Useful Lemmas}\label{app:lemmas}

We first show that the cumulative KL divergence from the best reference distribution to the predicted distribution is less than the predictor's worst-case regret.
\begin{lemma}\label{lemma:kl}
For a sequential predictor $\hat{p}_i$ with worst case regret $M(n)$, a collection observations $(x^n,y^n,z^n)$, and any distribution from the reference class $p \in \refclass_n$:
\begin{equation}
\sum_{i=1}^n \kl{p_i}{\hat{p}_i}\le M(n)
\end{equation}
\end{lemma}

\begin{proof}
\begin{align*}
\sum_{i=1}^n \kl{p_i}{\hat{p}_i}
&= \sum_{i=1}^n \sum_{x\in\mc{X}} p_i(x)
    \log \frac{p_i(x)}{\hat{p}_i(x)} \\
&\le \sum_{i=1}^n
    \left[ \sup_{x\in\mc{X}}
    \log \frac{p_i(x)}{\hat{p}_i(x)} \right]
    \sum_{x\in\mc{X}} p_i(x) \\
&= \sum_{i=1}^n \sup_{x\in\mc{X}} r(\hat{p}_i,p_i,x) \\
&\le \sup_{x^n\in\mc{X}^n} \sum_{i=1}^n r(\hat{p}_i,p_i,x_i)\\
&\le \sup_{x^n\in\mc{X}^n} \sup_{p\in\refclass_n} \sum_{i=1}^n r(\hat{p}_i,p_i,x_i)\\
&\le M(n)
\end{align*} \end{proof}
Next, we bound the cumulative difference in expectation of a bounded function between the best reference distribution and sequential predictor.

\begin{lemma} \label{lemma:g_func}
For a sequential predictor $\hat{p}_i$ with worst case regret $M(n)\ge 1$, a collection observations $(x^n,y^n,z^n)$, cumulative loss minimizing distribution $p^*_i$, and a collection of functions $g_i:\mc{X}\rightarrow \mathbb{R}$ for $i=1,\dots,n$:
\begin{equation}
\sum_{i=1}^n \left| E_{p^*_i}[g_i(X)] -
    E_{\hat{p}_i}[g_i(X)] \right| \le
    \frac{\left|\left|\vec{c}_n\right|\right|_2}{\sqrt{2}}\sqrt{M(n)}
\end{equation}
\noindent where $\vec{c}_n = [c_1,\dots,c_n]$ is a vector with elements:
\begin{equation}
c_i = \sum_{x\in \mc{X}} \left| g_i(x) \right|
\end{equation}
\end{lemma}

\begin{proof}
\begin{align}
\sum_{i=1}^n \left| E_{p^*_i} [g_i(X)] -
    E_{\hat{p}_i}[g_i(X)] \right| &= \sum_{i=1}^n \left| \sum_{x\in\mc{X}}
    \left[ p^*_i(x) - \hat{p}_i(x) \right] g_i(x) \right| \nonumber \\
&\le \sum_{i=1}^n \sum_{x\in\mc{X}} \left| p^*_i(x) -
    \hat{p}_i(x) \right| \left| g_i(x) \right|
    \label{ref_triangle_eq}\\
&\le \sum_{i=1}^n \left[\sum_{x\in\mc{X}} \left| p^*_i(x) -
    \hat{p}_i(x) \right|\right]\left[ \sum_{x\in\mc{X}} \left| g_i(x) \right|\right]
    \label{ref_positivity}\\
&\le \sum_{i=1}^n
    \sqrt{\frac{1}{2}\kl{p^*_i}{\hat{p}_i}} \sum_{x\in\mc{X}} \left| g_i(x) \right|
    \label{ref_pinsker}\\
&= \frac{1}{\sqrt{2}} \sum_{i=1}^n
    c_i \sqrt{\kl{p^*_i}{\hat{p}_i}} \nonumber
\end{align}

\noindent where \eqref{ref_triangle_eq} uses the triangle inequality, \eqref{ref_positivity} follows from both terms of the sum being positive, and \eqref{ref_pinsker} uses Pinsker's inequality. Focusing on the sum, we define a vector $\vec{v}_n=[v_1,\dots,v_n]$ such that $v_i = \sqrt{\kl{p^*_i}{\hat{p}_i}}$ for $i=1,\dots,n$:
\begin{align}
\sum_{i=1}^n c_i\sqrt{\kl{p^*_i}{\hat{p}_i}}
&= \left|\vec{c}_n\cdot\vec{v}_n\right| \label{posabsval} \\
&\le \left|\left|\vec{c}\right|\right|_2 \left|\left|\vec{v}\right|\right|_2 \label{cauchy}\\
&= \left|\left|\vec{c}\right|\right|_2\left( \sum_{i=1}^n \kl{p^*_i}{\hat{p}_i}
    \right)^{\frac{1}{2}} \nonumber \\
&\le \left|\left|\vec{c}\right|\right|_2\sqrt{M(n)} \label{ref_kl_lemma}
\end{align}

\noindent where \eqref{posabsval} follows from the fact that $c_i\ge0$ and $v_i\ge0$ for all $i$, \eqref{cauchy} uses the Cauchy--Schwarz inequality and \eqref{ref_kl_lemma} uses Lemma \ref{lemma:kl} and the assumption that $M(n) \ge 1$. \end{proof}

\section{Proof of Propositions}\label{app:props}
\subsection{Proof of Proposition \ref{prop:te}}\label{app:te}
Using the definition of the causal measure, we get that the left hand side of \eqref{teprop} is:
\begin{align*}
\sum_{\hc} p(\hc)
\kl{\fxc{i}}{\fxr{i}} &=\sum_{\hc}
p(\hc) \sum_{x_i} \fxc{i}(x_i)\log
\frac{\fxc{i}(x_i)}{\fxr{i}(x_i)} \\
&=\sum_{\hc}
p(\hc) \sum_{x_i} p(x_i\mid \hc)\log
\frac{p(x_i\mid\hc)}{p(x_i\mid\hr)} \\
&=\sum_{\hc,x_i}
p(\hc,x_i)\log
\frac{p(x_i\mid\hc)}{p(x_i\mid\hr)} \\
&=E_{X^i,Y^{i-1},Z^{i-1}}\left[\log
\frac{p(X_i\mid X^{i-1},Y^{i-1},Z^{i-1})}{p(X_i\mid X^{i-1},Z^{i-1})}
\right] \\
&= \diyxzm
\end{align*}
\qed

\subsection{Proof of Proposition \ref{prop:di_bounds}}\label{app:di_bounds}
First we note that the lower bound holds simply by applying conditioning reduces entropy to the first term of the PDI rate as defined in \eqref{partialdirateent}. Formally, we have:
\begin{align*}
\bar{I}_P^{(k)}(Y\rightarrow X) &= \bar{H}^{(k-1)}(X\mid\mid Y) - \bar{H}^{(1)}(X\mid\mid Y) \\
&=\lim_{n\rightarrow \infty}\frac{1}{n} \left[ H(X^n\mid\mid Y^{n-k-1}) - H(X^n\mid\mid Y^{n-1})\right] \\
&\le\lim_{n\rightarrow \infty}\frac{1}{n} \left[ H(X^n) - H(X^n\mid\mid Y^{n-1}) \right] \\
&=\bar{H}(X) - \bar{H}^{(1)}(X\mid\mid Y)=\dirateyx
\end{align*}
Next, for the TDI upper bound, we note that for $k_2=d$, the bound was proven in \cite{kontoyiannis2016estimating}. For completeness we extend the proof for $k_2\ge d$:
\begin{align*}
\bar{I}_T^{(k)}(Y\rightarrow X) &= \lim_{n\rightarrow \infty} \frac{1}{n} \sum_{i=1}^n
I(X_i;Y^{i-1}_{i-k_2}\mid X_{i-k_2}^{i-1}) \\
&= \lim_{n\rightarrow \infty} \frac{1}{n} \sum_{i=1}^n
H(X_i \mid X_{i-k_2}^{i-1}) - H(X_i \mid X_{i-k_2}^{i-1},Y^{i-1}_{i-k_2})\\
&= \lim_{n\rightarrow \infty} \frac{1}{n} \sum_{i=1}^n
H(X_i \mid X_{i-k_2}^{i-1}) - H(X_i \mid X_{i-d}^{i-1},Y^{i-1}_{i-d}) \\
&\ge \lim_{n\rightarrow \infty} \frac{1}{n} \sum_{i=1}^n
H(X_i \mid X^{i-1}) - H(X_i \mid X_{i-d}^{i-1},Y^{i-1}_{i-d})\\
&=\bar{H}(X) - \bar{H}^{(1)}(X\mid\mid Y)=\dirateyx
\end{align*}
\noindent as was to be shown. \qed

\subsection{Proof of Proposition \ref{prop:sideinfo_regret}}\label{app:sideinfo_regret}
As in the statement of the proposition, let $L$ be the number of leaves in the CTW and $N$ be the total number of nodes in the tree. Define $p_e(x^n)$ to be the Dirichlet estimator introduced in \cite{krichevsky1981performance}, otherwise known as the KT-estimator. Then it is known that the worst case regret is given by \cite{tjalkens1993sequential}:
\begin{equation} \label{ktbound}
\sup_{x_n} \log\frac{p(x^n)}{p_e(x^n)} \le
\frac{\absval{\mc{X}}-1}{2}\log n
+\absval{\mc{X}}-1
\end{equation}

We next define the KT-tree estimator with side information $p_c(x^n\mid\mid y^n)$ as the estimator where, for each possible ``context'' $(x_{i-d}^{i-1},y_{i-d}^{i-1})$, a separate instance of a KT-estimator is maintained. Letting $\mc{L}\triangleq\{(x_{i-d}^{i-1},y_{i-d}^{i-1}): (x_{i-d}^{i-1},y_{i-d}^{i-1})\in \mc{X}^d\times \mc{Y}^d\}$ be the set of contexts (i.e. leaf nodes), we have that $\absval{\mc{L}}=L$. Defining $p_e^{(l)}(x^n\mid\mid y^n)\triangleq p_e^{(l)}(x^{(l)})$ to be the KT-estimator that assigns probabilities to $x^{(l)}\triangleq\{x_i:(x_{i-d}^{i-1},y_{i-d}^{i-1})=l\}$ for $l\in\mc{L}$ with $|x_i^{(l)}|\triangleq n_l$, we can derive the worst case regret of the KT-tree estimator with side information as follows:
\begin{align}
\sup_{x^n,y^n} \log\frac{p(x^n)}{p_c(x^n)}
&=
\sup_{x^n,y^n} \log\frac{p(x^n\mid\mid y^n)}{\prod_{l\in\mc{L}}p_e^{(l)}(x^n\mid\mid y^n)} \nonumber \\
&=
\sup_{x^n,y^n} \log\prod_{l\in\mc{L}}\frac{p(x^{(l)})}{p_e^{(l)}(x^{(l)})} \nonumber \\
&=
\sup_{x^n,y^n} \sum_{l\in\mc{L}} \log \frac{p(x^{(l)})}{p_e^{(l)}(x^{(l)})} \nonumber \\
&\le \sum_{{l\in\mc{L}}} \left( \frac{\absval{\mc{X}-1}}{2}\log n_l + \absval{\mc{X}}-1  \right) \label{usektbound} \\
&= \frac{L(\absval{\mc{X}}-1)}{2}\sum_{{l\in\mc{L}}} \frac{1}{L}\log n_l + L(\absval{\mc{X}}-1) \nonumber \\
&\le \frac{L(\absval{\mc{X}}-1)}{2}\log \sum_{{l\in\mc{L}}} \frac{n_l}{L} + L(\absval{\mc{X}}-1) \label{jensenineq} \\
&= \frac{L(\absval{\mc{X}}-1)}{2}\log \frac{n}{L} + L(\absval{\mc{X}}-1) \label{sumnl}
\end{align}

\noindent where \eqref{usektbound} follows from the bound in \eqref{ktbound}, \eqref{jensenineq} follows from Jensen's inequality, and \eqref{sumnl} follows from the fact that $\sum_l n_l = n$. We now define the set of all nodes to be $\mc{S}\triangleq\{(x_{i-k}^{i-1},y_{i-k}^{i-1}): (x_{i-k}^{i-1},y_{i-k}^{i-1})\in \mc{X}^k\times \mc{Y}^k,k=1,\dots,d\}$, with $|\mc{S}|=S$. Then, we can define a context tree by letting defining a probability $p_w^{(s)}(x^n\mid\mid y^n)$ for each node $s\in\mc{S}$ as follows:
\begin{equation}
p_w^{(s)}(x^n\mid\mid y^n) =
\begin{cases}
\frac{1}{2}p_e^{(s)}(x^n\mid\mid y^n)+\frac{1}{2}\prod_{s' \in \mc{X}\times\mc{Y}} p_w^{(s's)}(x^n\mid\mid y^n)
\ &s\notin \mc{L} \\
p_e^{(s)}(x^n\mid\mid y^n)
\ &s\in \mc{L}
\end{cases}
\end{equation}

\noindent where $s's = (x_{i-k-1}^{i-1},y_{i-k-1}^{i-1}) \in \mc{X}^{k+1}\times \mc{Y}^{k+1}$ represents a child node of $s=(x_{i-k}^{i-1},y_{i-k}^{i-1})$ with $s'=(x_{i-k-1},y_{i-k-1})$. Letting $\lambda$ be the root node of the tree (i.e. $s\lambda=s$), the CTW probability assignment is given by $p_w(x^n\mid\mid y^n)\triangleq f_w^{(\lambda)}(x^n\mid\mid y^n)$. This probability assignment may be recursively lower-bounded as:
\begin{align}
p_w(x^n\mid\mid y^n)
&= \frac{1}{2}p_e^{(\lambda)}(x^n\mid\mid y^n) + \frac{1}{2}\prod_{s \in \mc{X}\times\mc{Y}}p_w^{(s)}(x^n\mid\mid y^n) \nonumber \\
&\ge \frac{1}{2}\prod_{s \in \mc{X}\times\mc{Y}}p_w^{(s)}(x^n\mid\mid y^n) \nonumber \\
&\ge \frac{1}{2}\prod_{s \in \mc{X}\times\mc{Y}}\frac{1}{2}\prod_{s' \in \mc{X}\times\mc{Y}}p_w^{(s's)}(x^n\mid\mid y^n) \nonumber \\
&\ge \dots \nonumber \\
&\ge \frac{1}{2^S}\prod_{l \in \mc{L}}p_w^{(l)}(x^n\mid\mid y^n) \nonumber \\
&= \frac{1}{2^S}\prod_{l \in \mc{L}}p_e^{(l)}(x^n\mid\mid y^n) \nonumber \\
&= \frac{1}{2^S}p_c^{(l)}(x^n\mid\mid y^n) \nonumber.
\end{align}

\noindent Finally, we can consider the log-likelihood ratio of true probability and the CTW probability  in order to obtain a bound on the worst case regret of the CTW:
\begin{align*}
\sup_{x^n,y^n} \log \frac{p(x^n\mid\mid y^n)}{p_w(x^n\mid\mid y^n)}
&\le S + \log \frac{p(x^n\mid\mid y^n)}{p_c(x^n\mid\mid y^n)} \\
&\le S + \frac{L(\absval{\mc{X}}-1)}{2}\log \frac{n}{L} + L(\absval{\mc{X}}-1)
\end{align*}

\noindent as was to be shown. \qed
\section{Proof of Theorems}\label{app:thms}

\subsection{Proof of Theorem \ref{thm:main_result}}\label{app:main_result}
We begin by defining the functions:
\begin{equation*}
\hat{g}_i(X) \triangleq \log \frac{\estfxc{i}(X)}{\estfxr{i}(X)} \ \ \ \ \
g^*_i(X) \triangleq \log \frac{\optfxc{i}(X)}{\optfxr{i}(X)}.
\end{equation*}

\noindent Using the definition of the causal measure and KL-divergence:
\begin{align}
\sum_{i=1}^n \left| \estcyx(i) - \optcyx(i) \right|
&- \left|
E_{\optfxc{i}}\left[ \hat{g}_i(X)\right] -
E_{\estfxc{i}}\left[ \hat{g}_i(X)\right]
\right| \label{beginning} \\
&= \sum_{i=1}^n
\left|
E_{\optfxc{i}}\left[ g^*_i(X)\right] -
E_{\estfxc{i}}\left[ \hat{g}_i(X)\right]
\right| - \left|
E_{\optfxc{i}}\left[ \hat{g}_i(X)\right] -
E_{\estfxc{i}}\left[ \hat{g}_i(X)\right]
\right| \nonumber \\
&\le \sum_{i=1}^n
\bigg| \left|
E_{\optfxc{i}}\left[ g^*_i(X)\right] -
E_{\estfxc{i}}\left[ \hat{g}_i(X)\right]
\right| - \left|
E_{\optfxc{i}}\left[ \hat{g}_i(X)\right] -
E_{\estfxc{i}}\left[ \hat{g}_i(X)\right]
\right| \bigg| \label{absval} \\
&\le \sum_{i=1}^n
\bigg|
E_{\optfxc{i}}\left[ g^*_i(X)\right] -
E_{\estfxc{i}}\left[ \hat{g}_i(X)\right]-
E_{\optfxc{i}}\left[ \hat{g}_i(X)\right] +
E_{\estfxc{i}}\left[ \hat{g}_i(X)\right] \bigg| \label{rev_tri} \\
&= \sum_{i=1}^n \left|
E_{\optfxc{i}}\left[ g^*_i(X) - \hat{g}_i(X)\right] \right| \nonumber \\
&= \sum_{i=1}^n \left|
E_{\optfxc{i}}\left[
\log \frac{\optfxc{i}(X)}{\estfxc{i}(X)}
- \log \frac{\optfxr{i}(X)}{\estfxr{i}(X)}
\right] \right| \nonumber \\
&\le \sum_{i=1}^n \left|
\kl{\optfxc{i}}{\estfxc{i}} \right|+
\left|
E_{\optfxc{i}}\left[
\log \frac{\optfxr{i}(X)}{\estfxr{i}(X)}
\right]
\right| \label{tri} \\
&\le M^{(c)}(n) + M^{(r)}(n) \label{regretlemma}
\end{align}

\noindent where \eqref{absval} follows from the properties of absolute value, \eqref{rev_tri} follows from the reverse triangle inequality, \eqref{tri} follows from the triangle inequality, and \eqref{regretlemma} follows from non-negativity of the KL-divergence, Lemma \ref{lemma:kl}, and Assumption \ref{assumption:referencekl}. Moving the second term of \eqref{beginning} to the other side of the inequality yields:
\begin{align}
\sum_{i=1}^n \left| \estcyx(i) - \optcyx(i) \right| & \le M^{(c)}(n) + M^{(r)}(n) +
\sum_{i=1}^n \left|
E_{\optfxc{i}}\left[ \hat{g}_i(X)\right] -
E_{\estfxc{i}}\left[ \hat{g}_i(X)\right]
\right| \nonumber \\
& \le M^{(c)}(n) + M^{(r)}(n) +
\frac{\left|\left|\vec{c}_n\right|\right|_2}{\sqrt{2}}\sqrt{M^{(c)}(n)}
\label{final_eq}
\end{align}

\noindent where \eqref{final_eq} follows from Lemma \ref{lemma:g_func}. This concludes the proof. \qed
\subsection{Proof of Theorem \ref{thm:dsep}}\label{app:dsepproof}
The first statement of the theorem follows trivially from the removal of $Y^{i-1}_{i-d}$ from $p(X_i \mid X^{i-1}_{i-d},Y^{i-1}_{i-d},Z^{i-1}_{i-d})$. Moving on, we will show that if $I(Y_j ; Y_k \mid X^i,Z^i)=0$ for all $j< k \le i$, $X$ is conditionally Markov of order at most $2d$ given $Z$. Note that:
\begin{align}
p(X_i \mid X^{i-1},Z^{i-1})
&=\sum_{y_{i-d}^{i-1}}
p(X_i \mid X^{i-1},y_{i-d}^{i-1},Z^{i-1})
\prod_{j=i-d}^{i-1}p(y_j \mid X^{i-1},Z^{i-1}) \label{indep} \\
&=\sum_{y_{i-d}^{i-1}}
p(X_i \mid X^{i-1}_{i-d},y_{i-d}^{i-1},Z^{i-1}_{i-d})
\prod_{j=i-d}^{i-1}p(y_j \mid X_{j-d}^{i-1},Z^{i-1}_{j-d}) \label{jmarkov} \\
&=\sum_{y_{i-d}^{i-1}}
p(X_i \mid X^{i-1}_{i-2d},y_{i-d}^{i-1},Z^{i-1}_{i-2d})
\prod_{j=i-d}^{i-1}p(y_j \mid X_{i-2d}^{i-1},Z^{i-1}_{i-2d}) \label{addish_condish}\\
&= p(X_i \mid X_{i-2d}^{i-1},Z^{i-1}_{i-2d}) \nonumber
\end{align}

\noindent where \eqref{indep} follows from the chain rule and that $y_{i-d}^{i-1}$ are conditionally independent given $(X^{i-1},Z^{i-1})$, \eqref{jmarkov} follows from the joint Markovicity of $X$ and $Y$ and the conditional independence of $y_{i-d}^{i-1}$, and \eqref{addish_condish} follows from the conditional independence of the past and the future given the present for Markov processes.

Next we will show that if there is some $j < k \le i$ such that $I(Y_j ; Y_k \mid X^i,Z^i)>0$, then there is no positive integer $l$ such that $(X_{i-l}^{i-1},Z_{i-l}^{i-1})$ d-separates $(X^{i-l-1},Z^{i-l-1})$ from $X_i$. To do this, we first note that $(X^i,Z^i)$ does not d-separate $Y_j$ and $Y_k$, because if it did, they would be conditionally independent. As such, when performing the d-separation algorithm given by Algorithm \ref{alg:dsep}, $Y_j$ and $Y_k$ will be connected by an undirected edge after completing step 4. Furthermore, if we let $\tau_1 = k-j$, then by the joint stationarity of $(X,Y,Z)$, \emph{every} $Y_i$ will be connected to $Y_{i-\tau_1}$ at the end of step 4. Furthermore, we know that $I(Y^n\ra X^n\mid Z^n)>0$ implies that for some $q\le m$, there is a directed edge from $Y_q$ to $X_m$. Letting $\tau_2 = m-q$, we know from the joint stationarity of $(X,Y,Z)$ that for every $X_i$, there is an incoming directed edge from $Y_{i-\tau_2}$. As such, at the end of step 4, every $X_i$ will be part of an undirected path connecting $Y_{i-\tau_2}$, $Y_{i-\tau_2-\tau_1}$, $Y_{i-\tau_2-2\tau_1}$, $\dots$. Thus, for any $l\ge 1$ this path can be followed $r$ steps such that $r\tau_1>d$. Then we know that $Y_{i-\tau_2-r\tau_1}$ is connected via an undirected edge to $X_{i-\tau_2-r\tau_1+\tau_2}=X_{i-r\tau_1}$. Recalling that in step 3 of the d-separation algorithm, $(X_{i-l}^{i-1},Z_{i-l}^{i-1})$ have been removed from the graph, we note that since $i-r\tau_1<i-l$, $X_{i-r\tau_1}$ is in the graph. Thus, there is an undirected path connecting $X_{r\tau_1} \in X^{i-l-1}$ and $X_i$, which implies that $(X_{i-l}^{i-1},Z_{i-l}^{i-1})$ does not d-separate $(X^{i-l-1},Z^{i-l-1})$ and $X_i$ for any $l$.\qed

\subsection{Proof of Theorem \ref{thm:necessary}}\label{app:zeromeasproof}

We will show that the statement holds for a fixed $l$, noting that a countably infinite union of measure zero sets has measure zero. First note that, if $X$ is conditionally $l$-Markov given $Z$, then for any $x_{i-l-1}^{i-1},x'_{i-l-1}\in\calx$ and $z_{i-l-1}^{i-1},z'_{i-l-1}\in\calz$ the following equality must hold:
\begin{equation}\label{constraint}
    p(x_i \mid x_{i-l-1}^{i-1},z_{i-l-1}^{i-1})=p(x_i \mid \tilde{x}_{i-l-1}^{i-1},\tilde{z}_{i-l-1}^{i-1})
\end{equation}
\noindent where we define $\tilde{x}_{i-l-d}^{i-1}\triangleq \{x_{i-l}^{i-1},x'_{i-l-1}\}$ and $\tilde{z}_{i-l-1}^{i-1}\triangleq \{z_{i-l}^{i-1},z'_{i-l-1}\}$. Define $\theta^{S_i}_{X,Y,Z}\triangleq p(S_i\mid X_{i-d}^{i-1},Y_{i-d}^{i-1},Z_{i-d}^{i-1})$ for $S\in\{X,Y,Z\}$ and $\theta\triangleq\{\theta^a_b:a\in\calx\cup\caly\cup\calz,b\in\calx^d\times\caly^d\times\calz^d\}$. We will demonstrate that the equation given by \eqref{constraint} amounts to solving a polynomial function of the parameters $\theta$. It is shown in \cite{okamoto1973distinctness} that the set of solutions to a non-trivial polynomial (i.e. one that is not solved by all of $\R^N$) will have Lebesgue measure zero with respect to $\R^N$. Focusing on the left side of \eqref{constraint}, we see that:
\begin{align}
p(x_i \mid x_{i-l-1}^{i-1},z_{i-l-1}^{i-1})
&=\sum_{\ym{l-1}^{i-1}} \tabp{x_i}{x,y,z}p(\ym{l-1}^{i-1}\mid \xm{l-1}^{i-1},\zm{l-1}^{i-1}) \nonumber \\
&=\sum_{\ym{l-1}^{i-1}} \tabp{x_i}{x,y,z}\frac{p(\xm{l-1}^{i-1},\ym{l-1}^{i-1}, \zm{l-1}^{i-1})}
{p( \xm{l-1}^{i-1},\zm{l-1}^{i-1})} \nonumber \\
&=\frac
{\sum_{\ym{l-1}^{i-1}} \tabp{x_i}{x,y,z}\pi(\xm{l-1},\ym{l-1}, \zm{l-1})
\prod_{j=1}^{l}\tabp{(x,y,z)_{i-j}}{x,y,z}}
{\sum_{\tilde{y}_{i-l-1}^{i-1}}\pi(\xm{l-1},\tilde{y}_{i-l-1}, \zm{l-1})
\prod_{j=1}^{l}\tabp{(x,\tilde{y},z)_{i-j}}{x,\tilde{y},z}}\label{hipi}
\end{align}
\noindent where $\pi:|\calx|\times|\caly|\times|\calz|\ra [0,1]$ is the invariant distribution and $\tabp{(x,y,z)_{i}}{x,y,z}\triangleq\tabp{x_i}{x,y,z}\tabp{y_i}{x,y,z}\tabp{z_i}{x,y,z}$. Next, define a matrix $A\in\R^{|\calx||\caly||\calz|\times|\calx||\caly||\calz|}$ containing the transition probabilities, i.e. $A_{j,k}=\tabp{R_k}{R_j}$ where $R$ is some enumeration over the $|\calx||\caly||\calz|$ possible values taken by $(X,Y,Z)$. Then we can represent $\pi$ in vector form $\pi\in[0,1]^{|\calx||\caly||\calz|}$ as the unique solution to $\pi=\pi A$. Let $\tilde{\pi}$ be an arbitrary vector satisfying $(A^T-I)\tilde{\pi}=0$, and note that for any $\tilde{\pi}$ there is a constant $C$ such that $C\tilde{\pi}=\pi$. Such a vector $\tilde{\pi}$ can be found by performing Gauss-Jordan elimination on $(A^T-I)$, and as a result, each element $\tilde{\pi}_j$ can be written as fractions of polynomial functions of $\theta$. Replacing $\pi$ with $C\tilde{\pi}$ in its functional form $\tilde{\pi}:|\calx|\times|\caly|\times|\calz|\ra \R$ in \eqref{hipi} we see that $C$ cancels in the numerator and denominator and thus each side of \eqref{constraint} can be written entirely as fractions of polynomial functions of $\theta$. Next, repeat the process on the right hand side of \eqref{constraint} with $\tilde{x}_{i-l-d}^{i-1}$ and $\tilde{z}_{i-l-d}^{i-1}$. Then, for any term that appears as a fraction, we can multiply both sides of \eqref{constraint} by the denominator and repeat until \eqref{constraint} is a polynomial function of $\theta$. Finally, we note that the polynomial given by \eqref{constraint} is trivial only if \emph{every} process is a solution which can be shown not to be the case by a straightforward counterexample.\qed

\subsection{Proof of Theorem \ref{thm:markov2}}
We will first show that $\tilde{X}$ is ($d+k$)-Markov, i.e. $\tilde{X}_i \perp \tilde{X}^{i-k-d-1} \mid \tilde{X}_{i-k-d}^{i-1}$. Note that the distribution of $\tilde{X}_i$ given $\tilde{X}^{i-1}$ may be written as:
\begin{align}
p(X_i,Y_{i-k+1} \mid  X^{i-1},Y^{i-k})
&= \sum_{y^{i-1}_{i-k+2}} p(X_i,Y_{i-k+1},y^{i-1}_{i-k+2} \mid X^{i-1},Y^{i-k})
\nonumber \\
&\triangleq \sum_{y^{i-1}_{i-k+2}} p(X_i,\tilde{Y}^{i-1}_{i-k+1} \mid X^{i-1},Y^{i-k})
\nonumber \\
&= \sum_{y^{i-1}_{i-k+2}} p(X_i \mid X^{i-1},\tilde{Y}^{i-1})
p(\tilde{Y}^{i-1}_{i-k+1} \mid X^{i-1},Y^{i-k})
\nonumber \\
&= \sum_{y^{i-1}_{i-k+2}} p(X_i \mid X^{i-1}_{i-k-d},\tilde{Y}^{i-1}_{i-k-d})p(\tilde{Y}^{i-1}_{i-k+1} \mid X^{i-1}_{i-k-d},Y^{i-k}_{i-k-d})
\label{ytil_markov} \\
&= \sum_{y^{i-1}_{i-k+2}} p(X_i,\tilde{Y}^{i-1}_{i-k+1} \mid X^{i-1}_{i-k-d},Y^{i-k}_{i-k-d})
\nonumber \\
&= p(X_i,Y_{i-k+1} \mid X^{i-1}_{i-k-d},Y^{i-k}_{i-k-d})
\nonumber
\end{align}

\noindent where, for ease of notation, we have defined $\tilde{Y}_j = y_j$ if $i-k+2\le j \le i-1$ and $\tilde{Y}_j = Y_j$ otherwise, and \eqref{ytil_markov} follows from the joint Markovicity of $X$ and $Y$ and:
\begin{align*}
p(\tilde{Y}^{i-1}_{i-k+1} \mid X^{i-1},Y^{i-k})
&= \prod_{j=i-k+1}^{i-1} p(\tilde{Y}_j \mid X^{i-1},\tilde{Y}^{j-1}) \\
&= \prod_{j=i-k+1}^{i-1} p(\tilde{Y}_j \mid X^{i-1}_{i-k-d},\tilde{Y}^{j-1}_{i-k-d}) \\
&= p(\tilde{Y}^{i-1}_{i-k+1} \mid X^{i-1}_{i-k-d},Y^{i-k}_{i-k-d})
\end{align*}

\noindent where we define $\tilde{Y}_a^b = \emptyset$ when $b<a$. This proves the Markovicity of $\tilde{X}$. To get the equality given by \eqref{xtildemarkov}, we simply take the sum over $Y_{i-k+1}$ in the above equations.

Next, we will show that $\tilde{X}$ is irreducible. We note that the possible states of $\tilde{X}$ may be a subset of the possible states of $(X,Y)$, i.e. $\tilde{\mc{X}}\subset\mc{X}\times\mc{Y}$. Each state $\tilde{x}\in\tilde{\mc{X}}$ occurs as a result of visiting a state $(x_{i-k+1},y_{i-k+1})$ followed by $(x_{i},y_i)$ after $k-1$ steps. Given that $(X,Y)$ is irreducible, every state $(x_{i-k+1},y_{i-k+1})\in\mc{X}\times\mc{Y}$ can be visited from any state $(x_{i},y_{i})\in\mc{X}\times\mc{Y}$. As a result, every state in $\tilde{x}\in\tilde{\mc{X}}$ can be visited from any other state $\tilde{x}'\in\tilde{\mc{X}}$. Therefore, $\tilde{X}$ is irreducible.

Lastly, we will show that if $(X,Y)$ is aperiodic then $\tilde{X}$ is also aperiodic. Note that for any state $\tilde{x}_i = (X_i=a,Y_{i-k+1}=b)\in\tilde{\mc{X}}$, we know there exists $c\in\mc{X}$ and $d\in\mc{Y}$ such that $p(X_{i-k+1}=c,Y_{i-k+1}=b,X_i=a,Y_i=d)>0$. By the aperiodicity we know that the greatest common divisor of the set of $\tau$ such that:
$$p(X_{i-k+1}=c,Y_{i-k+1}=b,X_{i-k+1+\tau}=c,Y_{i-k+1+\tau}=b)>0$$
\noindent is one. As a result, the same is true of $\tau$ such that:
\begin{align*}
0<&
p(X_{i-k+1}=c,Y_{i-k+1}=b,X_i=a,Y_i=d,X_{i-k+1+\tau}=c,Y_{i-k+1+\tau}=b,X_{i+\tau}=a,Y_{i+\tau}=d) \\
\le& p(Y_{i-k+1}=b,Y_{i-k+1+\tau}=b,X_i=a,X_{i+\tau}=a) \\
=& p(\tilde{x}_i,\tilde{x}_{i+\tau})
\end{align*}

\noindent implying that $\tilde{X}$ is aperiodic.
\qed

\subsection{Proof of Theorem \ref{thm:pdi_est}}
Let the estimate of the partial DI rate be given by:
\begin{equation}
\hat{I}^{(k)}_{P,n}(Y\rightarrow X) \triangleq
\frac{1}{n}\sum_{i=1}^n \kl{\estfxc{i}}{\estfxk{i}}.
\end{equation}

\noindent Then the theorem states that $\hat{I}^{(k)}_{P,n}$ converges to $\bar{I}_P^{(k)}$ almost surely. Following the proof of Theorem 3 in \cite{jiao2013universal}, decompose the estimate as:
\begin{equation*}
\hat{I}^{(k)}_{P,n}(Y\rightarrow X)
= \frac{1}{n}\sum_{i=1}^n \sum_{x_i} \estfxc{i}(x_i)
\log\frac{1}{\estfxk{i}(x_i)}
- \frac{1}{n}\sum_{i=1}^n \sum_{x_i} \estfxc{i}(x_i)
\log\frac{1}{\estfxc{i}(x_i)}.
\end{equation*}

\noindent It was shown in \cite{jiao2013universal} that the second term on the right hand side of the above equation converges to $\bar{H}^{(1)}(X\mid\mid Y)$ almost surely. Next, define the quantity:
\begin{equation}
F_n^{(k)}\triangleq \frac{1}{n}\sum_{i=1}^n \sum_{x_i} \estfxc{i}(x_i)
\log\frac{1}{\estfxk{i}(x_i)}.
\end{equation}

\noindent Then it remains to be shown that $F_n^{(k)}$ converges to $\bar{H}^{(k)}(X\mid\mid Y)$ almost surely. Next, define $R_n^{(k)}$ and $S_n^{(k)}$ as:
\begin{align*}
R_n^{(k)}
\triangleq&
\frac{1}{n} \sum_{i=1}^n \left[
\sum_{x_i}
\fxc{i}(x_i) \log \fxp{i}{k}(x_i) -
\estfxc{i}(x_i) \log \estfxk{i}(x_i) \right]\\
S_n^{(k)}
\triangleq&
-\frac{1}{n} \sum_{i=1}^n \sum_{x_i}
\fxc{i}(x_i) \log \fxp{i}{k}(x_i) -
\bar{H}^{(k)}(X\mid\mid Y)
\end{align*}

\noindent and note that $F_n^{(k)}-\bar{H}^{(k)}(X\mid\mid Y) = R_n^{(k)}+S_n^{(k)}$. As such, all that remains to be shown is that $R_n^{(k)}$ and $S_n^{(k)}$ converge to zero almost surely. It is shown in Lemma 2 of \cite{jiao2013universal} that the CTW probability assignment $\estfxc{i}(x_i)$ converges to $\fxc{i}(x_i)$ almost surely if $(X,Y)$ is a stationary irreducible aperiodic finite-alphabet Markov process. We showed in Theorem \ref{thm:markov2} that this condition implies that the process $\tilde{X}$ with $\tilde{X}_i\triangleq (X_i,Y_{i-k+1})$ is also a stationary aperiodic finite-alphabet Markov process and thus  $\estfxk{i}(x_i)$ converges to $\fxp{i}{k}(x_i)$ almost surely as well. As a result, we see that the bracketed term in $R_n^{(k)}$ converges to zero almost surely as $i$ tends to infinity. Furthermore, since $R_n^{(k)}$ is the Ces\'aro mean of the bracketed term, it too converges to zero almost surely.

To show that $S_n^{(k)}$ converges to zero, first define the first term as:
\begin{align*}
G_i^{(k)} &\triangleq
-\sum_{x_i}
\fxc{i}(x_i) \log \fxp{i}{k}(x_i) \\
&=-\sum_{x_i} p(x_i\mid x^{i-1},y^{i-1}) \log p(x_i\mid x^{i-1},y^{i-k}) \\
&=-\sum_{x_i} p(x_i\mid x^{i-1}_{i-k-d},y^{i-1}_{i-k-d}) \log p(x_i\mid x^{i-1}_{i-k-d},y^{i-k}_{i-k-d}) \\
&\triangleq g(x^{i-1}_{i-k-d},y^{i-1}_{i-k-d})
\end{align*}

Then, from Breiman's generalized ergodic theorem \cite{breiman1957individual}, it follows that the following equality holds almost surely:
\begin{equation*}
\lim_{n\rightarrow\infty}\frac{1}{n}\sum_{i=1}^n g(x^{i-1}_{i-k-d},y^{i-1}_{i-k-d}) = E\left[g(X^{-1}_{-k-d},Y^{-1}_{-k-d})\right].
\end{equation*}

\noindent Finally, using the law of iterated expectation, we note that:
\begin{align*}
E[g(X^{-1}_{-k-d},Y^{-1}_{-k-d})]
=&E\Bigg[-\sum_{x_0}p(x_0\mid X^{-1}_{-k-d},Y^{-1}_{-k-d}) \log p(x_0\mid X^{-1}_{-k-d},Y^{-k}_{-k-d})\Bigg] \\
=&E\Bigg[E\Bigg[
-\sum_{x_0}p(x_0\mid X^{-1}_{-k-d},Y^{-1}_{-k-d}) \log p(x_0\mid X^{-1}_{-k-d},Y^{-k}_{-k-d}) \bigg| X^{-1}_{-k-d},Y^{-k}_{-k-d}\Bigg]\Bigg] \\
=&E\Bigg[-\sum_{x_0}p(x_0\mid X^{-1}_{-k-d},Y^{-k}_{-k-d}) \log p(x_0\mid X^{-1}_{-k-d},Y^{-k}_{-k-d})\Bigg] \\
=&\bar{H}^{(k)}(X\mid\mid Y)
\end{align*}

\noindent Thus, we conclude that:
\begin{align*}
\lim_{n\rightarrow\infty} S_n^{(k)} &=
\lim_{n\rightarrow\infty} \frac{1}{n}\sum_{i=1}^{n}G_i^{(k)}-\bar{H}^{(k)}(X\mid\mid Y) = 0 \ p-a.s.
\end{align*}

\noindent as was to be shown. \qed

\section*{Acknowledgment}
The authors would like to thank Will Chapman for insightful conversations and the anonymous reviewers for their suggestions.

\printbibliography

\begin{IEEEbiography}
[{\includegraphics[width=1in,height=1.25in,clip,keepaspectratio]{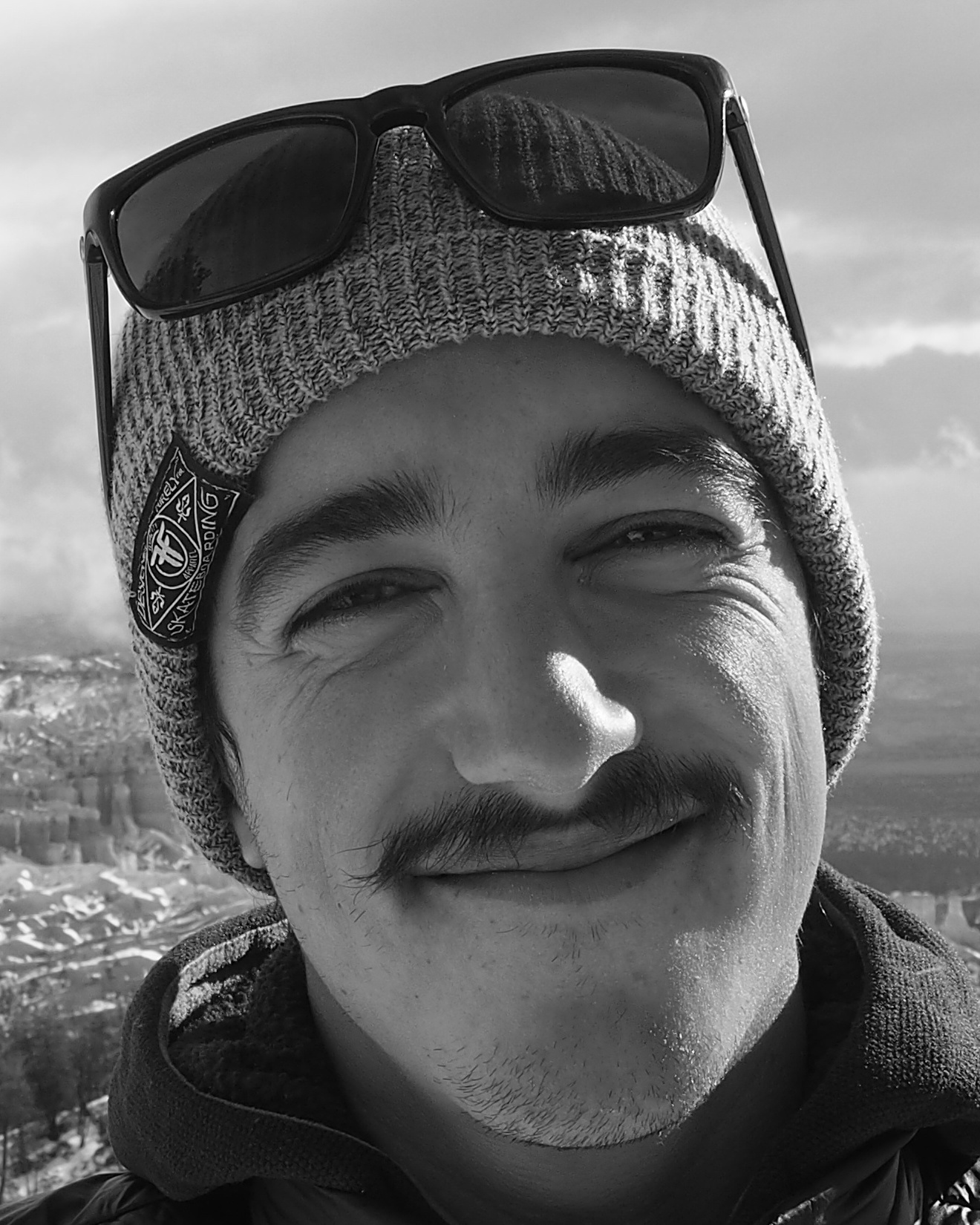}}]
{Gabriel Schamberg}
received his B.S. in computer engineering in 2012 and M.S. in electrical engineering in 2016, both from the University of California, San Diego (UCSD). From 2012 to 2014 he worked as a software developer for NKI Engineering in San Diego, California. In 2014 he began his graduate studies and was awarded the Jacobs Fellowship from the Department of Electrical and Computer Engineering at UCSD. Since 2014 he has worked in the Neural Interaction Lab at UCSD and is currently working toward a Ph.D in electrical engineering.
\end{IEEEbiography}
\begin{IEEEbiography}[{\includegraphics[width=1in,height=1.25in,clip,keepaspectratio]{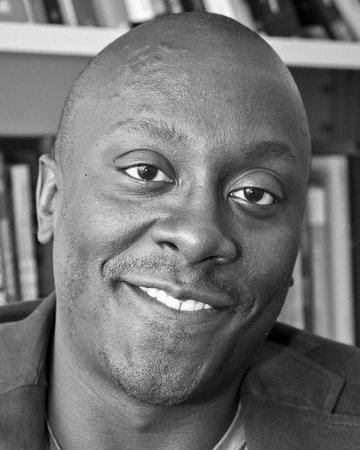}}]
{Todd P. Coleman}
received the Ph.D. degree in electrical engineering from the Massachusetts Institute of Technology (MIT), Cambridge, MA, USA, in 2005. He was a Postdoctoral Scholar in neuroscience at MIT and Massachusetts General Hospital (MGH) during the 2005-2006 academic year. He was an Assistant Professor in Electrical and Computer Engineering as well as Neuroscience at the University of Illinois from 2006 to 2011. He is currently a Professor in Bioengineering and Principal Investigator of the Neural Interaction Laboratory at the University of California San Diego, La Jolla, CA, USA. His research interests lie at the intersection of applied mathematics, bio-electronics, neuroscience, and medicine. He was named a Gilbreth Lecturer for the National Academy of Engineering in 2015 and a Fellow of the American Institute for Medical and Biological Engineering.
\end{IEEEbiography}

\end{document}